\documentclass[11pt,letterpaper]{article}
\usepackage[T1]{fontenc}

\title{Complexity of Unambiguous Problems in $\Sig{2}$}

\date{}

\author{
Matan Gilboa$^{1}$, 
Paul W. Goldberg$^{1}$, 
Elias Koutsoupias$^{1}$, 
Noam Nisan$^{2}$ \\
\\
$^{1}$University of Oxford \\
$^{2}$Hebrew University of Jerusalem
}

\usepackage{amsthm} 
\usepackage{amsmath}
\usepackage{amssymb}
\usepackage{graphicx}
\usepackage{dsfont}
\usepackage{appendix}

\usepackage{algorithm}
\usepackage{algpseudocode}
\usepackage[margin=1in,paper=letterpaper]{geometry}
\usepackage{multirow}
\usepackage{natbib}
\usepackage{setspace}
\usepackage{tikz}
\usetikzlibrary{positioning, shapes, fit, backgrounds, calc, arrows.meta}

\usepackage[unicode, hidelinks]{hyperref}
\usepackage[nameinlink]{cleveref}

\usepackage[most]{tcolorbox}

\newtcolorbox[auto counter, number within=section]{wbox}[1][]{%
  colback=white,
  colframe=black,
  fonttitle=\bfseries,
  title=Problem~\thetcbcounter: #1,
  boxrule=0.5pt,
  boxsep=1pt,        
  left=2pt,          
  right=2pt,
  top=2pt,
  bottom=2pt,
  before skip=4pt,   
  after skip=4pt,    
  arc=0pt            
}

\newtheorem{theorem}{Theorem}[section]
\newtheorem{lemma}[theorem]{Lemma}

\newtheorem{prop}{Proposition}
\newtheorem{observation}{Observation}
\newtheorem{definition}[theorem]{Definition}

\crefname{wbox}{Problem}{Problems}
\Crefname{wbox}{Problem}{Problems}
\crefname{prop}{proposition}{propositions}   
\Crefname{prop}{Proposition}{Propositions}  
\crefname{lemma}{lemma}{lemmas}   
\Crefname{lemma}{Lemma}{Lemmas}
\crefname{observation}{observation}{observations}   
\Crefname{observation}{Observation}{Observations}

\usepackage{xcolor}
\definecolor{mygray}{RGB}{218,215,203}

\usepackage{todonotes}
\presetkeys{todonotes}{inline}{}
\newcommand{\subqed}{%
  {\renewcommand{\qedsymbol}{$\triangle$}\qedhere}%
}
\newcommand{\EX}{\ensuremath{\mathbb{E}}}
\newcommand{\NN}{\ensuremath{\mathbb{N}}}
\newcommand{\RR}{\ensuremath{\mathbb{R}}}
\newcommand{\Prob}{\ensuremath{\mathbb{P}}}
\DeclareMathOperator{\sgn}{sgn} 

\newcommand{\fCC}{\ensuremath{\mathfrak{C}}} 
\newcommand{\CC}{\ensuremath{\mathcal{C}}} 
\newcommand{\VV}{\ensuremath{\mathcal{V}}} 
\newcommand{\fXX}{\ensuremath{\mathcal{X}}} 
\newcommand{\fYY}{\ensuremath{\mathcal{Y}}} 
\newcommand{\XX}{\ensuremath{\mathtt{x}}} 
\newcommand{\YY}{\ensuremath{\mathtt{y}}} 
\newcommand{\ZZ}{\ensuremath{\mathtt{z}}} 
\newcommand{\WW}{\ensuremath{\mathtt{w}}} 

\newcommand{\bs}{\setminus}
\newcommand{\reduce}{\ensuremath{\leq_P}}
\newcommand{\bits}[1]{\ensuremath{\{0,1\}^{#1}}}
\newcommand{\concat}{\mathbin\Vert}

\newcommand{\complexityclass}[1]{\ensuremath{\bf{#1}}}
\newcommand{\NP}{\complexityclass{NP}}
\newcommand{\coNP}{\complexityclass{coNP}}
\newcommand{\Pclass}{\complexityclass{P}}
\newcommand{\UP}{\complexityclass{UP}}
\newcommand{\UAP}{\complexityclass{UAP}}
\newcommand{\DifP}{\complexityclass{D^P}}
\newcommand{\PNP}{\complexityclass{\Pclass^{NP}}}
\newcommand{\PCW}{\complexityclass{PCW}}
\newcommand{\PTW}{\complexityclass{PTW}}
\newcommand{\PMA}{\complexityclass{PMA}}
\newcommand{\USig}[1]{\complexityclass{U\Sigma_{#1}^P}}
\newcommand{\SymP}{\complexityclass{S_2^P}}
\newcommand{\ZPPNP}{\complexityclass{ZPP^{NP}}}
\newcommand{\MA}{\complexityclass{MA}}
\newcommand{\AM}{\complexityclass{AM}}
\newcommand{\BPP}{\complexityclass{BPP}}
\newcommand{\LOPClass}{\complexityclass{L_2^P}}
\newcommand{\coAM}{\complexityclass{coAM}}

\newcommand{\Piclass}[1]{\complexityclass{\Pi_{#1}^P}}
\newcommand{\Dif}[1]{\complexityclass{D_{#1}^P}}
\newcommand{\Sig}[1]{\complexityclass{\Sigma_{#1}^P}}

\newcommand{\SAT}{\textsc{SAT}}
\newcommand{\USat}{\textsc{USAT}}
\newcommand{\UnSat}{\textsc{UNSAT}}
\newcommand{\UQSAT}{\textsc{U\textnormal{\ensuremath{\exists\forall}}-Sat}}

\newcommand{\UTSP}{\textsc{TSP-Unique-Opt}}
\newcommand{\UOPT}{\textsc{Ckt-Unique-Opt}}
\newcommand{\CKTConsensus}{\textsc{Ckt-Consensus}}
\newcommand{\SDomStrat}{\textsc{Strong-Dominant-Strategy}}
\newcommand{\CktWinThresh}{\textsc{Ckt-Winner-Threshold}}

\newcommand{\CKTCondorcet}{\textsc{Ckt-Condorcet}}
\newcommand{\kCKTCondorcet}[1]{\textsc{Ckt-Condorcet}[\mathit{#1}]}
\newcommand{\ASHG}{\textsc{Ashg-Strong-Popularity}}
\newcommand{\GDice}{\textsc{Graph-Dice}}
\newcommand{\CKTDice}{\textsc{Ckt-Dice}}
\newcommand{\StrictCKTDice}{\textsc{Strict-Ckt-Dice}}

\newcommand{\WTournSource}{\textsc{Weak-Tournament-Source}}
\newcommand{\TournSource}{\textsc{Tournament-Source}}
\newcommand{\MultiTournSource}{\textsc{Multi-Weak-Tournament-Source}}
\newcommand{\LOP}{\textsc{LOP}}

\newcommand{\EMaj}{\textsc{Edge-Majority}}

\newcommand{\kEMajBal}[1]{\textsc{Edge-Majority-Balanced}[\mathit{#1}]}
\newcommand{\EMajBal}{\textsc{Edge-Majority-Balanced}}
\newcommand{\EMajSet}{\textsc{Edge-Majority-Set}}

\newcommand{\WDomStrat}{\textsc{WDom-Strategy}}

\newcommand{\CKTUVAL}{\textsc{Ckt-Unique-Value}}
\newcommand{\CKTPareto}{\textsc{2-Ckt-Pareto}}
\newcommand{\QSAT}{\textsc{$\exists\forall$-Sat}}
\newcommand{\coQSAT}{\textsc{$\forall\exists$-Sat}}

\begin{document}

\maketitle

\begin{abstract}
Various practical problems within the class \Sig{2} possess an {\em unambiguity} property, meaning that yes-instances correspond with a {\em unique} witness. The semantic class containing all unambiguous $\Sig{2}$ problems is denoted $\USig{2}$.
Examples include the existence of (1) a dominating strategy in a game, (2) a Condorcet winner, 
(3) a strongly popular partition in hedonic games, and 
(4) a winner (source) in a tournament.
The computational complexity of unambiguous problems is not well understood, leaving many questions unresolved.
We address this gap in a broad complexity-theoretic sense; our main contributions consist of the following.

\begin{itemize}
    \item We identify three {\em syntactic} subclasses of $\USig{2}$ associated with general properties of problems that guarantee uniqueness: Polynomial Tournament Winner (\PTW{}), Polynomial Condorcet Winner (\PCW{}), and Polynomial Majority Argument (\PMA{}). 
    \item We establish complexity upper and lower bounds for our proposed classes. In particular, we show that they are all contained in \SymP{} and are thus significantly easier than the immediate \Sig{2} upper bound.
    \item We characterize the complexity of various practical problems using this framework.
\end{itemize}
\end{abstract}
\newpage
\tableofcontents
\newpage

\section{Introduction}
\label{sec:introduction}
Problems in \NP{} ask about the existence of solutions that can be efficiently checked, and problems in \Sig{2} ask about the existence of solutions that can be efficiently checked in \coNP{}. 
Such a decision problem is called \textit{unambiguous} if it is guaranteed to have at most one witness.
Unambiguous problems within \NP{} (that is, those in \UP{}) that are not known to belong to \Pclass{} are quite scarce. Notable examples include Prime Factorization, Parity Games, and Mean Payoff Games \cite{jurdzinski1998deciding}. 
Beyond \NP{}, however, unambiguous problems become considerably more diverse.
One level higher in the polynomial hierarchy, we encounter many natural examples, and our focus therefore lies on problems contained in the semantic\footnote{The distinction between syntactic complexity classes, which have complete problems, and semantic classes, which are unlikely to have complete problems, is discussed in \cite{Papa94a}. Syntactic classes are associated with machines having some recognizable structure, whereas semantic classes have problems solved by machines in which some (possibly undecidable) assumption is needed about how the machine behaves (for \UP{}, the assumption is that a nondeterministic Turing machine has at most one accepting path for any input).} class \USig{2}, consisting of all unambiguous problems in \Sig{2}.
The classification of their computational complexity poses a new and interesting challenge: 
Unambiguity becomes an obstacle, since polynomial-time reductions---our main tool for such tasks---tend to be parsimonious, meaning that they preserve the number of solutions between instances \cite{hartmanis1976isomorphisms,valiant1974reduction,simon1977difference}.
Hence, it is not always clear how to reduce a problem that may admit multiple solutions to one that cannot.
This paper is a comprehensive exploration of \USig{2}.

We first motivate the study of \USig{2} with some concrete problems.
In voting, a Condorcet winner \cite{Cond85a} is a candidate $w$ that is preferred over any other candidate $x$ by a majority of voters, if voters' preferences are limited to $w$ versus $x$. Clearly, there can be at most one Condorcet winner. The existence of a Condorcet winner may be hard to determine in settings where the number of candidates is exponential in the input length. One problem where this occurs is \ASHG{}, the existence problem of a strongly popular partition in an additively separable hedonic game.
In this problem, outcomes (the candidates) are partitions of a set of agents (the voters) into {\em coalitions}, and any agent values a partition based on the sum of values that they assign to the agents in their own coalition.
This problem was shown in \cite{BrBu20a} to be \coNP{}-hard, but its best known upper bound is the trivial \Sig{2} bound.    

As another example, we introduce a problem \GDice{} based on the notion of \emph{intransitive dice} \cite{gardner1970paradox} --- we seek a die that has a higher probability of winning than losing against any other die, aiming to roll the higher number.

A further unambiguous problem is determining the existence of a strongly dominant strategy in a game with an exponential strategy space (\SDomStrat{}). Deciding whether a given player has a strongly dominant strategy is syntactically in \Sig{2}, and there cannot be more than one such strategy.

Naturally, a central goal in this work is to characterize the computational complexity of the above problems, as well as related ones, noting that known complexity classes may not suffice to capture their structure. 
Perhaps the only well-known class in $\Sig{2}$ (and presumably not in \NP{}) associated primarily with unambiguous problems is \PNP{} \cite{stockmeyer1976polynomial}. Papadimitriou \cite{papadimitriou1984complexity} showed that determining whether an instance of the traveling salesman problem admits a uniquely optimal solution (\UTSP{}) is complete for \PNP{}. Various other problems, typically involving a solution space that induces a weak order, have subsequently received a similar characterization \cite{krentel1986complexity}.
The structure of the underlying order allows the usage of an \NP{} oracle to perform binary search over the solution space, yielding containment in \PNP{}.
However, many of our motivating problems lack this convenient structure.  
We thus aim to isolate the core features of \USig{2} problems by identifying the combinatorial principles that give rise to unambiguity.

This idea results in the introduction of three complexity classes based on corresponding Boolean-circuit problems. The first two classes, denoted Polynomial Tournament Winner (\PTW{}) and Polynomial Condorcet Winner (\PCW{}), manage to capture {\em all problems mentioned above}. 
\PTW{} is based on the problem of determining the existence of a source in a tournament graph.
The above problems are all tournament problems of sorts, where we are able to efficiently compare candidate solutions and we seek a winner beating all others (e.g., in \ASHG{}, a partition $\pi$ beats a partition $\pi'$ if $\pi$ is more popular than $\pi'$ in a majority vote among the agents). 
There can be at most one such winner because of the anti-symmetric nature of the ``beat'' relation defining the tournament ($\YY$ cannot beat $\XX$ if $\XX$ beats $\YY$).
\PCW{} is a subclass of \PTW{} which captures tournament-winner problems where comparison between vertices is obtained via a majority vote among a polynomial number of voters. 
The third class, Polynomial Majority Argument (\PMA{}), relies on a completely different combinatorial principle. Rather than obtaining unambiguity from the anti-symmetry of the ``beat'' relation, we derive it from a global condition that is not locally checkable. 
More concretely, we syntactically construct a bipartite graph $G=(X,Y,E)$ whose number of edges is strictly less than $2|Y|$, and we ask whether there exists a vertex in $X$ adjacent to all vertices in $Y$. 
Such a vertex would consume a majority of the edges, and thus yes-instances admit just one witness.
We find the definition of \PMA{} intriguing, as it exposes where $\USig{2}$ (and the class \SymP{}, as later discussed) seems stronger than \PTW{}.

Using this framework, we classify the complexity of the problems discussed, and others. Specifically, we show that \SDomStrat{}, \CKTConsensus{}, and the Condorcet winner problem with two voters, are all \PNP{}-complete, along with another problem denoted \CktWinThresh{}; and that \GDice{} is \PCW{}-complete. 
In \cite{gilboa2026StrongPopularity} it is further shown that \ASHG{} is \PCW{}-complete.

The classes \PTW{}, \PCW{}, and \PMA{} seem robust, in that they capture at least all \USig{2} problems we have been able to come up with. 
To better understand the complexity of these problems, we relate them to pre-existing complexity classes. 
Are these unambiguous classes significantly easier than their immediate superclass \Sig{2}? 
Conversely, can we find any computational lower bounds for them?
We answer these questions in the affirmative. Specifically, we prove that $
\PNP{} \subseteq \PCW{} \subseteq \PTW{} \subseteq \SymP{} $, $\BPP\subseteq\MA\subseteq\LOPClass{}\subseteq\PTW$, and $ \coNP{} \subseteq \PMA{} \subseteq \SymP{}\cap\coAM{}$ (see \Cref{fig:complexity_classes_arrows}).

To show that $\PTW\subseteq \SymP{}$, we build on an interesting graph-theoretic result by Erd{\H{o}}s. 
The class \SymP{}, introduced independently by \cite{russell1998symmetric} and \cite{canetti1996more}, captures problems in which there are symmetrically-defined polynomial-size witnesses for yes-instances and no-instances (formally defined in \Cref{sec:prelims}). 
\SymP{} is shown in \cite{cai2007s2p} to be a subclass of \ZPPNP{}. 
This means that, with randomization, we have that \PTW{} collapses to \PNP{} and \PMA{} become a subset of \PNP{}.
We thus show that these broad classes of unambiguous problems have much stronger upper bounds on their complexity than \Sig{2}.

The class \LOPClass{} is based on the problem of finding the minimal element of a linear order \cite{korten2024strong}. It seems to retain most of the interesting properties of \SymP{}, such as the strengthening of the Karp-Lipton collapse \cite{KarpLiptonCollapse1980,cai2007s2p,LOPBounds2025}, the containment of \MA{}, \BPP{}, and the range avoidance problem, which tightly relates to circuit complexity \cite{kleinberg2021total,korten2022hardest,ren2022range,chen2023range,korten2024strong} (see \cite{korten2025range} for a recent survey on range avoidance).
\PTW{} finds itself in the intriguing position in-between \LOPClass{} and \SymP{}, thus inheriting the above properties of \LOPClass{}.
Furthermore, this implies an oracle separation between \PNP{} and \PTW{}, as there is an oracle relative to which \PNP{} is not a subset of \BPP{} \cite{muchnik1996general}.
It is also interesting to note that both \LOPClass{} and \PCW{} generalize \PNP{}, are generalized by \PTW{}, but are not trivially comparable to each other.

Next, we address the question of how pivotal is the role of unambiguity in the computational complexity of problems. 
We look into some of the unambiguous problems in this paper, and make small changes to eliminate this property. 
We show that these variations turn out to be \Sig{2}-complete, significantly harder than their unambiguous counterparts. 
We exemplify this with problems we denote \CKTUVAL{} and \CKTPareto{}, which are variations of \UOPT{} and \CKTConsensus{} (defined in \Cref{sec:PNP}), respectively, and with a variant of hedonic games.

Finally, we consider an unambiguous problem which seems to lie outside of \Sig{2}, namely determining the existence of a weakly-dominant strategy in a game. We show it is computationally equivalent to the two-quantified version of \USat{} (\UQSAT{}). 
We further prove that they are positioned between \Piclass{2} and \Dif{2}, supporting the view that the hardness of unambiguous problems does not typically reflect the full potential of their immediate upper-bounding class---in this case, $\Sig{3}$.

The rest of this paper is organized as follows.
Following a brief account of related literature, in \Cref{sec:prelims} we introduce some relevant notation and definitions. \Cref{sec:PTW,sec:PCW,sec:PMA} introduce the classes \PTW{}, \PCW{}, and \PMA{} respectively, discuss their relation with other complexity classes, and problems that are captured by them.
In \Cref{sec:PNP} we exhibit several \PNP{}-complete unambiguous problems. 
\Cref{sec:Sig2} considers ambiguous variants of unambiguous problems, and shows they are \Sig{2}-complete. 
\Cref{sec:Beyond_sig2} discusses the complexity of the existence of a weakly-dominant strategy and \UQSAT{}.
All proofs are deferred to the appendix.
\begin{figure}[h]
\centering
\begin{tikzpicture}[
    node distance=1cm and 0.8cm,
    class/.style={align=center},
    arrow/.style={-{Stealth[length=3mm]}, thick}
]

\node[class] (coNP) {\coNP{}};
\node[class, above=0.3cm of coNP] (BPP) {\BPP{}};
\node[class, right=1cm of BPP] (MA) {\MA{}};
\node[class, right=1cm of MA] (LOP) {\LOPClass{}};
\node[class, right=1cm of coNP] (PNP) {\PNP{}};
\node[class, red, right=1cm of PNP] (PCW) {\PCW{}};
\node[class, red, right=1cm of LOP] (PTW) {\PTW{}};
\node[class, right=1.5cm of PTW] (S2P) {\SymP{}};
\node[class, right=1cm of S2P] (ZPPNP) {\ZPPNP{}};
\node[class, red, below right=0.3cm and 3.1cm of coNP] (PMA) {\PMA{}};
\node[class, right=11.5cm of coNP] (Sigma2) {\Sig{2}};
\node[class, right=5cm of PCW] (USigma2) {\USig{2}};
\node[class, right=5cm of PMA] (coAM) {\coAM{}};

\draw[arrow] (coNP) -- (PNP);
\draw[arrow] (BPP) -- (MA);
\draw[arrow] (MA) -- (LOP);
\draw[arrow] (PNP) -- (LOP);
\draw[arrow] (LOP) -- (PTW);
\draw[arrow] (PNP) -- (PCW);
\draw[arrow] (PCW) -- (PTW);
\draw[arrow] (PTW) -- (S2P);
\draw[arrow] (PMA) -- (USigma2);
\draw[arrow] (PTW) -- (USigma2);
\draw[arrow] (S2P) -- (ZPPNP);
\draw[arrow] (USigma2) -- (Sigma2);
\draw[arrow] (ZPPNP) -- (Sigma2);

\draw[arrow] (coNP) -- (PMA);
\draw[arrow] (PMA) -- (S2P);
\draw[arrow] (PMA) -- (coAM);
\draw[arrow] (coAM) -- (Sigma2);

\end{tikzpicture}
\caption{Our proposed complexity classes (in red) relative to known subclasses of \Sig{2}.}
\label{fig:complexity_classes_arrows}
\end{figure}
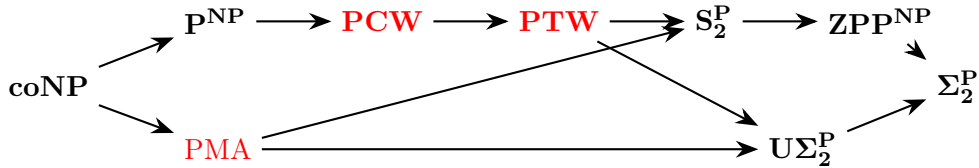
\subsection{Related Literature}
\label{sec:RL}
We conclude this introduction with a brief account of related studies, in addition to the works discussed throughout the paper.
A result of Stearns \cite{stearns1959voting} shows that, given a tournament $T$ on $n$ vertices, $\theta(\frac{n}{\log n})$ voters may be required to induce a majority graph that coincides with $T$. 
A related result of Erd\H{o}s \cite{erdos1964representation} shows that $\theta(\frac{n}{\log n})$ voters also suffice to induce any tournament.
The former result may hint at a separation between the classes \PTW{} and \PCW{}.

Another well-studied problem with the unambiguity property is \USat{}: Given a Boolean formula, does it admit a {\em unique} satisfying assignment? \cite{papadimitriou1982complexity} shows that \USat{} is contained in \DifP{}, and several papers observe that it is \coNP{}-hard. 
A proof for this can be found in \cite{blass1982unique}, where they also show an oracle relative to which the problem is complete for \DifP{}, and another one relative to which it is not. 
\USat{} seems easier than most of the problems discussed in this work, as \DifP{} is a subset of \PNP{}. See \Cref{sec:Beyond_sig2} for further discussion on \DifP{} and its generalization.

In \Cref{sec:PCW} we address computational aspects of intransitive dice. The literature typically approaches them from the perspectives of combinatorics, probability theory, and game theory.
Intransitive dice first appeared in \cite{gardner1970paradox}, which credits Bradley Efron for a set of three dice, known as Efron's dice, that form a winning cycle.
In \cite{akin2020generalized,schaefer2017balanced}, it is shown that every tournament on $n$ vertices can be induced by a set of dice using the ``beat'' relation. 
\cite{schaefer2017balanced} further show that for any $n,m\ge 3$, there exists a set of $n$ dice, each with $m$ faces, that form an $n$-cycle, a result later rediscovered by \cite{coelho2023central}. 
\cite{coelho2023central} study a random model where the dice's faces are random variables. 
\cite{hulko2019game} consider a game where two agents simultaneously construct an $n$-faced die, and then compare them; they show it admits a unique pure Nash equilibrium.
\cite{cornacchia2020intransitive} study the probability of obtaining an intransitive tournament when randomly constructing four dice.

Numerous works have discussed unambiguous computation. 
The class \UP{} was introduced in \cite{valiant1976relative}, and has since then been tied with foundational questions in cryptography. It is shown in \cite{grollmann1988complexity} that one-way functions exist if and only if $\Pclass{}\ne\UP{}$.
The famous Valiant-Vazirani theorem establishes that \complexityclass{Promise-\UP{}} is as hard as \NP{}, up to randomization \cite{VV85}.
In \cite{lange1994unambiguous}, three seemingly different hierarchies of unambiguous classes were introduced and characterized using Boolean circuits of exponential size with suitable restrictions on gate behavior.
\cite{niedermeier1998unambiguous} provides an alternative characterization of \UP{} and introduces the class \UAP{}, defined via local functions on the nodes of a Turing machine’s computation tree.
Our definition of \USig{2} aligns with the one found in \cite{FortnowYamakami96}, where they oracle-separate $\PNP{}$ from $\USig{2}\cap \Piclass{2}$ (a result subsumed by the above observation that $\PNP{}$ is oracle-separated from $\PTW{}$). Interestingly, their proof involves a language bearing some resemblance to our definition of \PMA{}.
\section{Preliminaries}
\label{sec:prelims}
Given a directed graph $G=(V,E)$ and $v_1,v_2\in V$, we say $v_1$ beats $v_2$ if $(v_1,v_2)\in E$ while $(v_2,v_1)\notin E$ (namely, an outgoing edge is considered a win in a contest between the vertices). If both edges are present or neither of them are, we normally say there is a tie between the vertices. 
We say $G$ is a \emph{weak tournament} if for every $v_1,v_2\in V$, at least one of the edges $(v_1,v_2)$ or $(v_2,v_1)$ is present in $E$. We say it is a \emph{tournament} if exactly one of those edges is present in $E$ for every pair of vertices.

We let $\Sigma$ be a fixed, finite alphabet; without loss of generality $\Sigma=\{0,1\}$. $\Sigma^*$ is then the set of all finite strings over $\{0,1\}$.
We define $\USig{k}$ as follows.
\begin{definition}
(\cite{FortnowYamakami96}).
Let $k\ge 1$. A language $L\subseteq\Sigma^*$ is in $\USig{k}$ if there is a polynomial-time Turing machine $M$ and a polynomial $p$ such that for all $w\in \Sigma^*$ we have:
\begin{itemize}
    \item $w\in L\iff \exists \XX_1 \;\forall \XX_2\;\exists \XX_3\dots Q\XX_k\; M(w,\XX_1,\dots,\XX_k)=1$, and
    \item there exists at most one $\XX_1$ such that $\forall \XX_2\;\exists \XX_3\dots Q\XX_k\; M(w,\XX_1,\dots,\XX_k)=1$
\end{itemize}
where $Q$ alternates between $\exists$ and $\forall$, and for all $i$ we have $\XX_i\in\Sigma^*$ and $|\XX_i|\leq p(|w|)$.
\end{definition}
Using the dual characterization of the polynomial hierarchy (\cite{stockmeyer1976polynomial,wrathall1976complete}), it may be shown that an equivalent definition is $\USig{k}=\UP{}^{\Sig{k-1}}$ for all $k\ge 1$.
We informally say a problem is \emph{unambiguous} if some argument shows it can have at most one witness.
Given problems $A$ and $B$, $A\reduce B$ means $A$ is many-one reducible to $B$ in polynomial time.
$\NN$ is the set of positive integers, and $\NN_0:=\NN\cup\{0\}$.
Given two bit strings $\XX$ and $\YY$ we write $\XX\YY$ or $\XX\concat\YY$ for their concatenation. 
We often interpret bit strings as binary numbers, and use operations like $\XX<\YY$.

\begin{definition}
\label{def:S2P}
(\cite{russell1998symmetric,canetti1996more}). A language $L\subseteq\Sigma^*$ is in $\SymP$ if there is a polynomial-time predicate $P(\cdot,\cdot,\cdot)$ such that for all $w\in\Sigma^*$ we have:
\begin{itemize}
    \item $w\in L\implies \exists \XX \;\forall \YY\;\; P(w,\XX,\YY)=1$, and
    \item $w\notin L\implies \exists \YY\; \forall \XX\;\; P(w,\XX,\YY)=0$
\end{itemize}
where $|\XX|$ and $|\YY|$ are bounded by a polynomial function of $|w|$.
\end{definition}

\begin{definition}
\label{def:L2P}
(\cite{korten2024strong}). A language $L\subseteq\Sigma^*$ is in $\LOPClass{}$ if there is a polynomial-time predicate $P(\cdot,\cdot,\cdot)$ and a polynomial $p$ such that for all $w\in\Sigma^*$, $P(w,\cdot,\cdot)$ defines a strict total order on $\bits{p(n)}$ whose minimal element $a$ has $a_1=L(w)$ ($a_1$ being the first bit of $a$).
Minimality of $a$ is defined by $P(w,a,b)=1$ for all $b$.
\end{definition}

It is straightforward that \USig{2} and $\SymP{}$ are closed under polynomial-time many-one reductions. The same holds for the classes \PTW{}, \PCW{}, and \PMA{} defined in \Cref{sec:PTW,sec:PCW,sec:PMA}. Hence, to show that a class $\complexityclass{A}$ is contained in a class $\complexityclass{C}\in\{\USig{2},\SymP{},\PTW{},\PCW{},\PMA{}\}$, it suffices to exhibit an $\complexityclass{A}$-complete language $L$ such that $L\in \complexityclass{C}$.

\section{The Class Polynomial Tournament Winner (\PTW{})}
\label{sec:PTW}
In this section, we introduce a class denoted Polynomial Tournament Winner (\PTW{}). \PTW{} is a subclass of \USig{2} which aims to describe problems of a ``tournament'' nature, where we seek a winner that beats all other participants under a polynomial-time computable criterion. This class contains all unambiguous \Sig{2} problems discussed in this paper, apart from \EMaj{} and its variations, discussed in \Cref{sec:PMA}.
We define \PTW{} via a problem we call \WTournSource{}. \WTournSource{} is the problem of determining the existence of a source in an exponentially large weak tournament. It is described through a Boolean circuit which compares pairs of vertices in polynomial time, and the question is whether there exists a vertex that beats all other vertices, which immediately places the problem in \Sig{2}.

\begin{wbox}
    \WTournSource\\
    \textbf{Input:} Boolean circuit $\CC\colon \bits{2n}\to\bits{}$.\\
    \textbf{Question:}
    $\exists \XX\in\bits{n} \;\; \forall \YY\in\bits{n} \text{ with } \YY\neq \XX\;\; \CC(\XX\concat \YY)=1 \land \CC(\YY\concat \XX)=0$?
\end{wbox}

\begin{definition}
\label{def:PTW}
The class Polynomial Tournament Winner (\PTW{}) is defined by
\[\PTW{}=\{L\subseteq\Sigma^*\colon L\reduce \WTournSource{}\}.\]
\end{definition}

$\CC$ induces a weak tournament where, for all distinct $\XX,\YY$, the edge $(\XX,\YY)$ is present if and only if $\CC(\XX\concat \YY)\ge\CC(\YY\concat \XX)$. In particular, both $(\XX,\YY)$ and $(\YY,\XX)$ are present if $\CC(\XX\concat \YY)=\CC(\YY\concat \XX)$.
Thus, a solution corresponds to a source in the induced tournament.
Indeed, \WTournSource{} is unambiguous, and so $\PTW{}\subseteq\USig{2}$. This is because given two vertices, either they tie against each other or one of them beats the other, and so either way they cannot both be sources.
   
An important result of this section is that $\PTW{}\subseteq \SymP{}$, which implies that $\PTW{}\subseteq \ZPPNP{}$, since $\SymP{}\subseteq\ZPPNP{}$ \cite{cai2007s2p}. 
An interesting step along the way of our proof shows that finding a source in a weak tournament is computationally equivalent to finding one in a tournament (\TournSource{}), as shown in \Cref{thm:TournSource_PTW}.

\begin{wbox}
    \TournSource\\
    \textbf{Input:} Boolean circuit $\CC\colon \bits{2n}\to\bits{}$.\\
    \textbf{Question:}
    $\exists \XX\in\bits{n} \;\; \forall \YY\in\bits{n} \text{ with } \YY\neq \XX\;\; \big(\CC(\XX\concat \YY)=1 \land \CC(\YY\concat \XX)=0\big)\lor \big(\CC(\XX\concat \YY)=\CC(\YY\concat \XX) \land \XX>\YY\big)$?
\end{wbox}

$\CC$ induces a tournament where the edge $(x,y)$ is present if and only if $\big(\CC(\XX\concat \YY)=1 \land \CC(\YY\concat \XX)=0\big)\lor \big(\CC(\XX\concat \YY)=\CC(\YY\concat \XX) \land \XX>\YY\big)$ (we use their numeric values as a tie breaker to ensure it describes a tournament rather than a weak one). Once again, a solution corresponds to a source. The following two theorems establish the main results of the section.

\begin{theorem}
\label{thm:TournSource_PTW}
\TournSource{} is \PTW{}-complete.
\end{theorem}

\begin{theorem}
\label{thm:PTW_S2P}
$\PTW{} \subseteq \SymP{}$.
\end{theorem}

The containment of \PTW{} in \SymP{} provides a non-trivial upper bound for several problems, such as \ASHG{} and \GDice{}, discussed in \Cref{sec:PCW}. 
In \Cref{sec:PCW} we define the subclass \PCW{} of \PTW{} for which these problems are complete. We further show there that $\PNP{}\subseteq \PCW{}$, thus also providing a lower bound for \PTW{}.

We also obtain a distinct lower bound, relying on the work of \cite{korten2024strong}. They define the class \LOPClass{} based on the problem \textsc{Linear-Order-Principle} (\LOP{}), which essentially seeks the minimal element in a linear order. They show it contains both \PNP{} and \MA{}.
By interpreting linear orders as tournaments, we show that $\LOPClass{}\subseteq\PTW{}$.

\begin{theorem}
\label{thm:L2P_PTW}
$\LOPClass{} \subseteq \PTW{}$.
\end{theorem}

To illustrate the robustness of the definition of \PTW{}, we introduce a generalization of \WTournSource{} which we call \MultiTournSource{}, and show that it is \PTW{}-complete. It generalizes \WTournSource{} by adding a third input string $z$ for the circuit; ignoring this input recovers an instance of \WTournSource{}. 
Intuitively, each $z$ induces a weak tournament $T_z$, and we seek a vertex that serves as a source across all such tournaments.

\begin{wbox}
    \MultiTournSource\\
    \textbf{Input:} Boolean circuit $\CC\colon \bits{3n}\to\bits{}$.\\
    \textbf{Question:}
    $\exists \XX\in\bits{n} \;\; \forall \YY,\ZZ\in\bits{n} \text{ with } \YY\neq \XX\;\; \CC(\XX\concat \YY\concat \ZZ)=1 \land \CC(\YY\concat \XX\concat \ZZ)=0$?
\end{wbox}

\begin{theorem}
\label{thm:coNP_source_PTW}
\MultiTournSource{} is \PTW{}-complete.
\end{theorem}

\section{The Class Polynomial Condorcet Winner (\PCW{})}
\label{sec:PCW}
In this section, we introduce a subclass of \PTW{} termed Polynomial Condorcet Winner (\PCW{}), inspired by the notion of a Condorcet winner \cite{Cond85a}. 
This class captures tournament problems in which the underlying tournament is induced by pairwise majority comparisons among candidates, determined by a polynomial-size set of voters. 
Each voter is represented by a Boolean circuit that succinctly encodes the voter’s preferences.
Given a bit string representing a candidate, each circuit outputs a numerical value, thereby inducing a weak total order over the preferences.
The question is whether there exists a candidate who defeats every other candidate in a pairwise majority vote---namely, a Condorcet winner. 
To formalize this notion, we define the problem \CKTCondorcet{}, and we define \PCW{} as the class of problems polynomial-time reducible to it.

\begin{wbox}
    \CKTCondorcet\\
    \textbf{Input:} A sequence of Boolean circuits $\fCC=\langle\CC_1,\dots,\CC_m\rangle$, each with $n$ inputs and $n$ outputs.\\
    \textbf{Question:}
    $\exists \XX\in\bits{n} \;\; \forall \YY\in\bits{n} \text{ with } \YY\neq \XX\;\; \sum_{i=1}^m\big(\sgn(\CC_i(\XX)-\CC_i(\YY))\big)>0$?\\
    A solution is called a \emph{Condorcet string}.
\end{wbox}

\begin{definition}
The class Polynomial Condorcet Winner (\PCW{}) is defined by
\[\PCW{}=\{L\subseteq\Sigma^*\colon L\reduce \CKTCondorcet{}\}.\]
\end{definition}

We establish lower and upper computational bounds on \PCW{}, formalized as follows.

\begin{theorem}
\label{thm:PCW_PTW}
$\PCW{}\subseteq\PTW{}$.
\end{theorem}

\begin{theorem}
\label{thm:PNP_PCW}
$\PNP{}\subseteq\PCW{}$.
\end{theorem}

Our bounds can be extended to the parameterized version of \CKTCondorcet{}, where the number of input circuits is a fixed parameter $k$.
For $k=1$ and $k=2$, the resulting problems are shown to be \PNP{}-complete in \Cref{sec:PNP} (when $k=1$ we obtain the problem \UOPT{}, defined in \Cref{sec:PNP}).
For any $k\ge 3$, we do not know the problem's exact complexity, but given the results for $k=1$ and $k=2$ it is naturally \PNP{}-hard and in \PCW{}.

\begin{wbox}
    $\kCKTCondorcet{k}$\\
    \textbf{Parameter:} A fixed, positive integer $k$.\\
    \textbf{Input:} Sequence of $k$ Boolean circuits $\fCC=\langle\CC_1,\dots,\CC_k\rangle$, each with $n$ inputs and $n$ outputs.\\
    \textbf{Question:}
    $\exists \XX\in\bits{n} \;\; \forall \YY\in\bits{n} \text{ with } \YY\neq \XX\;\; \sum_{i=1}^k\big(\sgn(\CC_i(\XX)-\CC_i(\YY))\big)>0$?
\end{wbox}

\begin{theorem}
\label{thm:Condorcet_parameterized}
Fix $k\geq 1$. We have that:
\begin{enumerate}
    \item For all $k'>k$, we have that $\kCKTCondorcet{k}\reduce \kCKTCondorcet{k'}$.\label{thm:Condorcet_parameterized:kk'}
    \item $\kCKTCondorcet{k}\in\PCW{}$\label{thm:Condorcet_parameterized:PCW}.
    \item $\kCKTCondorcet{k}$ is \PNP{}-hard.\label{thm:Condorcet_parameterized:PNP}
\end{enumerate} 
\end{theorem}

Unlike \PNP{}, whose complete problems typically have a weak total order over the solution space, \PCW{} provides us with a tool to examine problems with intransitive solution spaces\footnote{While intransitivity hints as to why problems are difficult to place in $\PNP{}$, it does not immediately imply they are not in \PNP{}. For instance, the solutions of $\kCKTCondorcet{2}$ do not necessarily induce a total order, but we show it is \PNP{}-complete (\Cref{thm:pnp_results}). 
To see why, consider two voters $v_1$ and $v_2$ and four candidates $a$, $b$, $c$, and $d$. 
Voter $v_1$ assigns the values $(4,4,3,5)$ and voter $v_2$ assigns $(2,1,5,3)$ to $(a,b,c,d)$, respectively.
It may be verified that pairwise majority comparisons yield $a\succ b$, $b\sim c$, $c\sim d$, and $d\succ a$, showing intransitivity.\label{foot:two_cond_intransitive}}, as is the case for the Condorcet winner problem \cite{Cond85a}.

An appealing property of \PCW{} is that it has natural complete problems (not determined by arbitrary Boolean circuits).
The first of those is the existence of \emph{strongly popular} partitions in \emph{additively separable hedonic game} (ASHG).
In an ASHG, agents need to be partitioned into coalitions, and each agent's utility is the sum of her valuations for the other agents in her coalition. 
A partition is strongly popular if it wins a majority vote against any other partition \cite{Gard75a}. 
Determining the existence of such a partition was shown to be \coNP{}-hard in \cite{BrBu20a}.
The question of its precise complexity was raised in \cite{BrBu20a,BG24popularity}. 
It seems that \PCW{} is the correct class for this problem, as
\cite{gilboa2026StrongPopularity} shows that it is \PCW{}-complete.

\begin{theorem}
\label{thm:ashg_pop_pcw}
(\cite{gilboa2026StrongPopularity}).
\ASHG{} is \PCW{}-complete.
\end{theorem} 

It is not too difficult to show that \PCW{} contains the strong popularity existence question for any type of hedonic games where utilities are induced from cardinal valuations to individual agents, since each agent can be easily simulated by a Boolean circuit.
Thus, \PCW{} provides an upper bound for strong popularity in many hedonic-games variants, such as fractional and modified fractional hedonic games, that previously had no upper bound better than the trivial \Sig{2} bound.
We suspect that those problems are in fact \PCW{}-complete.

We proceed with another problem whose solutions are intransitive, which we term \GDice{}, and its Boolean-circuit generalization \CKTDice{}. 
The motivation for these problems originates in the well-studied mathematical concept of \emph{intransitive dice} (see \Cref{sec:RL}). Suppose we have a set of dice, each with $m$ faces, and on each face is a number. Each die has a probability of $\frac{1}{m}$ to roll each of its faces. We say die $x$ \emph{beats} die $y$ if the probability\footnote{In the literature often a probability greater than $\frac{1}{2}$ is required.} that $x$ rolls a higher number than $y$ is greater than the probability that $y$ rolls a higher number than $x$. 
Similarly to popularity, the dice game is also intransitive; that is, we may have a set of dice that beat each other in a cyclic manner.

Consider thus the problem of deciding whether there exists a die that beats all others in a given set of dice. 
When the set is small, this can be easily checked by iterating over all pairs of dice. 
However, when the set of dice is large this becomes an interesting computational problem. We introduce two formulations for this problem. In the more general formulation, \CKTDice{}, we let $n$-bit strings represent dice labels, and their faces are generated by the outputs of a collection of Boolean circuits whose input is the die label. The second formulation, \GDice{}, is more natural in that it is defined on a directed graph with weighted edges. Each partition $\pi$ of the vertices induces a die, where each vertex $v$ corresponds to a face; the value of that face is the weighted sum of outgoing edges from $v$ to other vertices in the set that $v$ belongs to by $\pi$.
We will show that, for both formulations, determining the existence of a winning die is \PCW{}-complete.

\begin{wbox}
    \CKTDice\\ 
    \textbf{Input:} Sequence of Boolean circuits $\fCC=\langle\CC_1,\dots,\CC_m\rangle$, each with $n$ inputs and $n'$ outputs.\\
    \textbf{Question:}
    $\exists \XX\in\bits{n} \;\; \forall \YY\in\bits{n} \text{ with } \YY\neq \XX\;\; \sum_{i=1}^m\sum_{j=1}^m\big(\sgn(\CC_i(\XX)-\CC_j(\YY))\big)>0$?
    
    We interpret $\CC_i(x)$ as the value on the $i$'th face of die $x$, as signed integers written in binary.
\end{wbox}

\begin{wbox}
    \GDice\\ 
    \textbf{Input:} Directed graph $G=(V,w)$, where $|V|=n$ and $w$ is an edge-weight function $w:V\times V\to \RR$ on all edges, including self loops.\\
    \textbf{Question:}
    $\exists \text{partition }\pi^* \;\; \forall \text{partition }\pi' \text{ with } \pi'\neq \pi^*\;\; \sum_{i=1}^n\sum_{j=1}^n\big(\sgn(s_i(\pi^*)-s_j(\pi'))\big)>0$?

    Where $\pi(i)$ denotes the set to which vertex $i$ belongs in the partition $\pi$, and $s_i(\pi):=\sum_{l\in\pi(i)}w_i(l)$.
\end{wbox}

\begin{theorem}
\label{thm:all_dice_PCW_complete}
\GDice{} and \CKTDice{} are \PCW{}-complete.
\end{theorem}

\section{The Class Polynomial Majority Argument (\PMA{})}
\label{sec:PMA}
In this section we introduce another subclass of \USig{2}, termed Polynomial Majority Argument (\PMA{}). \PMA{} captures unambiguous problems in \Sig{2} whose unambiguity follows from a majority argument, intuitively described as follows. 
We are given a directed bipartite graph $G=(X,Y,E)$ with edges from $X$ to $Y$, where the existence of edges can be efficiently queried by a Boolean circuit. Assume that $|E|<2|Y|$, a property that we will show can be syntactically enforced.
We then ask whether there exists a vertex $x\in X$ such that for every $y\in Y$, $(x,y)\in E$. 
Unambiguity follows since two vertices cannot both claim ownership for a majority of the edges.
Formally:

\begin{wbox}
\label{def:EMaj}
    \EMaj{}\\ 
    \textbf{Input:} Two Boolean circuits $\CC\colon \bits{n+m}\to\bits{m+1}$ and $\VV\colon \bits{m+1}\to\bits{n+m}$, and integer $k\ge 1$.\\
    \textbf{Question:}
    $\exists \XX\in\bits{n} \;\; \forall \YY\in\bits{m}\;\; \VV(\CC(\XX\concat \YY))=\XX \YY \;\land\; \CC(\XX\concat \YY) \ge k$?
\end{wbox}

\begin{definition}
\label{def:PMA}
The class Polynomial Majority Argument (\PMA{}) is defined by 
\[\PMA{}=\{L\subseteq\Sigma^*\colon L\reduce \EMaj{}\}.\]
\end{definition}

We interpret $\{l\in\bits{m+1}\colon l \ge k\}$ as the set of admissible edge-labels.
The circuit $\VV$ (``verifier'') ensures that edges are correctly mapped in the reverse direction---from the label $\CC(\XX\concat\YY)$ back to the vertex pair $(\XX,\YY)$. 
Since there are only $2^{m+1}-k$ admissible labels, this enforces an upper bound on the number edges.
Moreover, as evident from its syntactic form, \EMaj{} is in \Sig{2}.

\begin{theorem}
\label{thm:PMA_Usig2}
$\PMA{}\subseteq\USig{2}$.
\end{theorem}

We next establish computational bounds for \PMA{}.
Similarly to \PTW{}, we show that \PMA{} is a subset of \SymP{}, implying that the latter is an upper bound for all unambiguous \Sig{2} problems considered in this work.
We additionally prove that \PMA{} is contained in \coAM{}, which has no known comparison to \SymP{}.
As a lower bound, we prove that \PMA{} contains \coNP{}, providing evidence that this class indeed contains hard problems.

\begin{theorem}
\label{thm:PMA_bounds}
$\coNP{}\subseteq\PMA{}\subseteq\SymP{}\cap \coAM{}$.
\end{theorem}

We also consider several natural variants of \EMaj{}.
We define a parameterized variant denoted $\kEMajBal{k}$, in which the threshold of disqualified labels $k$ is a fixed parameter, and the bipartite graph is balanced ($|X|=|Y|$). 
Surprisingly, for every fixed $k\ge 1$ $\kEMajBal{k}$ remains \PMA{}-complete. 
This result naturally extends to variants imposing only one of these restrictions (either balance or fixed $k$).
Furthermore, if the input specifies an explicit set of disqualified labels instead of a threshold, the problem remains \PMA{}-complete.

\begin{wbox}
\label{def:kEMajBal}
    $\kEMajBal{k}$\\ 
    \textbf{Parameter:} A fixed integer $k\ge 1$.\\
    \textbf{Input:} Two Boolean circuits $\CC\colon\bits{2n}\to\bits{n+1}$ and $\VV\colon \bits{n+1}\to\bits{2n}$.\\
    \textbf{Question:}
    $\exists \XX\in\bits{n} \;\; \forall \YY\in\bits{n}\;\; \VV(\CC(\XX\concat \YY))=\XX \YY \;\land\; \CC(\XX\concat \YY) \ge k$?
\end{wbox}

\begin{wbox}
\label{def:EmajSet}
    \EMajSet{}\\ 
    \textbf{Input:} Two Boolean circuits $\CC\colon\bits{2n}\to\bits{n+1}$ and $\VV\colon \bits{n+1}\to\bits{2n}$, and a set $\emptyset\ne S\subseteq \NN_0$ given explicitly as a list of binary numbers.\\
    \textbf{Question:}
    $\exists \XX\in\bits{n} \;\; \forall \YY\in\bits{n}\;\; \VV(\CC(\XX\concat \YY))=\XX \YY \;\land\; \CC(\XX\concat \YY) \notin S$?
\end{wbox}

\begin{theorem}
\label{thm:Emaj_k_pma} 
Let $k\ge 1$. $\kEMajBal{k}$ is \PMA{}-complete.
\end{theorem}

\begin{theorem}
\label{thm:EmajSet_pma} 
\EMajSet{} is \PMA{}-complete.
\end{theorem}

We conjecture that neither \PMA{} nor \PTW{} is contained in the other, as their unambiguity relies on fundamentally distinct combinatorial principles.
\section{Unambiguous Problems Complete for \PNP{}}
\label{sec:PNP}

In this section, we consider several unambiguous problems that are complete for \PNP{}, the class of problems that are Turing reducible to \SAT{}. 
Since $\PNP{}\subseteq \PCW{}\subseteq\PTW{}$, these problems seem easier than other \PTW{} problems discussed above.
Specifically, we prove \PNP{}-completeness of $\kCKTCondorcet{2}$ (defined in \Cref{sec:PCW}), and the following three problems. 

\begin{wbox}
    \SDomStrat\\
    \textbf{Input:} Boolean circuit $\CC\colon\bits{2n}\to\bits{n}$.\\
    \textbf{Question:}
    $\exists \XX\in\bits{n} \;\; \forall \XX',\YY\in\bits{n} \text{ with } \XX\neq \XX' \;\; \CC(\XX,\YY)>\CC(\XX',\YY)$?
\end{wbox}

\begin{wbox}
    \CKTConsensus\\
    \textbf{Input:} Sequence of Boolean circuits $\fCC=\langle\CC_1,\dots,\CC_m\rangle$, each with $n$ inputs and $n$ outputs.\\
    \textbf{Question:}
    $\exists \XX\in\bits{n} \;\; \forall \XX'\in\bits{n} \text{ with } \XX'\neq \XX\;\; \forall i\in \{1,\dots m\}\;\; \CC_i(\XX)>\CC_i(\XX')$?
\end{wbox}

\begin{wbox}
    \CktWinThresh\\
    \textbf{Input:} A Boolean circuit $\CC\colon \bits{2n}\to\{0,1\}$.\\
    \textbf{Question:}
    $\exists \XX\in\bits{n} \;\; \forall \XX',\YY\in\bits{n} \text{ with } \XX'\neq \XX\;\; \CC(\XX,\YY)=1\land \big(\XX'>\XX\implies \CC(\XX',\YY)=0\big)$?
\end{wbox}

In \SDomStrat{}, we interpret $\CC$ as formulating a game where $\XX$ describes the strategy of Player~1, $\YY$ describes the strategies of the other players, and $\CC$'s output is the utility of Player~1. The question then is whether there exists a strongly dominant strategy for Player~1. We refer to \Cref{sec:Beyond_sig2} for results on finding a weak-dominant strategy in a game.

\CKTConsensus{} asks for a string achieving uniquely optimal value in all given circuits. This can be seen as a circuit-based generalization of \UTSP{}, where several sets of weights (each encoded by a different circuit) are given for the same graph and we seek a TSP tour (encoded by the input $\XX$) that is uniquely optimal with respect to all sets of weights.

\CktWinThresh{} asks whether there is an $\XX$ that is an all-winner (it satisfies $\psi$ regardless of $\YY$) but it is also a threshold, in that any string greater than $\XX$ is an all-loser (it does not satisfy $\psi$ regardless of $\YY$).
Notice that it differs from the other three problems in that the circuit's output is in $\{0,1\}$, whereas the other problems' circuits output strings. Thus, hardness for \CktWinThresh{} stems only from the size of the input space of $\psi$, while in the other problems it also relates to the size of the output space of the circuits, yielding different-flavored reductions.

$\kCKTCondorcet{2}$ is the Condorcet-winner existence problem with two voters. 
It is an interesting example of a problem that, despite having intransitive solutions (see \Cref{foot:two_cond_intransitive}) resembling the \PCW{}-complete problems discussed in \Cref{sec:PCW}, is included in \PNP{}.

The results of this section are summarized in the following theorem.

\begin{theorem}
\label{thm:pnp_results}
The problems \SDomStrat{}, \CKTConsensus{}, \CktWinThresh{}, and $\kCKTCondorcet{2}$ are \PNP{}-complete.
\end{theorem}

\section{Ambiguous Problem Variants}
\label{sec:Sig2}
In the previous sections, we have seen many examples of unambiguous problems in \Sig{2} that are relatively easy compared with the full potential of \Sig{2}.
In this section, we consider slight variations of previously discussed problems, which are ambiguous. We exhibit the idea that unambiguity seems to play an important role in diminishing the computational complexity of problems, by showing that their ambiguous variations yield \Sig{2}-complete problems.
One such example can be found in \cite{BG24popularity}, where it is shown that Weak-Popularity in additive hedonic games (a version of \ASHG{} defined with a weak inequality) is \Sig{2}-complete, in contrast to our result of \Cref{thm:ashg_pop_pcw}.

We provide two additional such results, showing \Sig{2}-completeness of the problems \CKTUVAL{} and \CKTPareto{}, defined below. 
Indeed, it may be observed that \CKTUVAL{} differs from \UOPT{} only in that $>$ becomes $\neq$ (instead of asking about the existence of a string obtaining a value \emph{greater} than any other string, we ask about obtaining a value \emph{different} than any other string); and \CKTPareto{} differs from \CKTConsensus{} (with 2 circuits) only in that the condition changes from $\land$ to $\lor$ (we seek a string which, against any other string, wins in either one of the circuits, rather than in both circuits). A solution for \CKTPareto{} is Pareto optimal in the sense that any string that challenges it would reduce the value of at least one circuit.

\begin{wbox}
    \CKTUVAL\\ 
    \textbf{Input:} Boolean circuit $\CC\colon\bits{n}\to\bits{n}$.\\
    \textbf{Question:} $\exists\XX\in\bits{n}\;\forall\YY\in\bits{n}\text{ with }\YY\neq\XX\;\; \CC(\XX)\neq\CC(\YY)$?
\end{wbox}

\begin{wbox}
    \CKTPareto\\ 
    \textbf{Input:} Two Boolean circuits $\CC_1\colon\bits{n}\to\bits{n}$ and $\CC_2\colon\bits{n}\to\bits{n}$.\\
    \textbf{Question:} $\exists\XX\in\bits{n}\;\forall\YY\in\bits{n}\text{ with }\YY\neq\XX\;\; \CC_1(\XX)>\CC_1(\YY)\lor\CC_2(\XX)>\CC_2(\YY)$?
\end{wbox}

\begin{theorem}
\label{thm:CKTUVAL_Sig2}
    \CKTUVAL{} is \Sig{2}-complete.
\end{theorem}

\begin{theorem}
\label{thm:CKTPareto_Sig2}
    \CKTPareto{} is \Sig{2}-complete.
\end{theorem}

\section{Beyond \Sig{2}: Weak Dominant Strategy and \UQSAT{}}
\label{sec:Beyond_sig2}
In this section, we turn our attention to the problem of determining the existence of a weakly-dominant strategy in a game. There, we seek a strategy $s_1^*$ for Player~1 such that for any other strategy $s_1'$, we have that $s_1^*$ is weakly better than $s_1'$ against any strategy profile of the other players, and is strictly better than $s_1'$ against \emph{some} strategy profile of the other players. We formalize this problem as follows.

\begin{wbox}
    \textsc{Weak-Dominant-Strategy} (\WDomStrat)\\
    \textbf{Input:} Boolean circuit $\CC\colon\bits{2n}\to\bits{n}$.\\
    \textbf{Question:}
    $\exists \XX\in\bits{n} \;\; \forall \XX',\YY\in\bits{n} \text{ with } \XX\neq \XX' \;\;\exists \YY'\in\bits{n} \;\; \CC(\XX\concat \YY)\geq\CC(\XX'\concat \YY) \land \CC(\XX\concat \YY')>\CC(\XX'\concat \YY')$?
\end{wbox}

As in \SDomStrat{}, we interpret $\CC$ as describing a game in which $\XX$ specifies the strategy of Player~1, $\YY$ specifies the joint strategies of the remaining players, and the circuit's output represents Player~1's utility.
An inspection of the problem's definition shows it is unambiguous. 
However, unlike all previously discussed problems, it seems improbable that it is included in \Sig{2}. 
Indeed, as suggested by the quantifier structure in its definition, its complexity seems to lie somewhere between the second and third levels of the polynomial hierarchy. 

Supporting this intuition, we first show that \WDomStrat{} is many-one equivalent to \UQSAT{} (``Is there a unique satisfying assignment in a two-quantified Boolean formula?'').
We then show that \UQSAT{} is \Piclass{2}-hard and contained in \Dif{2}, a class introduced in \cite{aleksandrowicz2017computational} that generalizes \DifP{} \cite{papadimitriou1982complexity}.

This characterization is analogous to known results for \USat{} (``Is there a unique satisfying assignment in a one-quantified Boolean formula?''), which is \coNP{}-hard and contained in \DifP{} \cite{blass1982unique, papadimitriou1982complexity}.
We now proceed to define a generalization of \DifP{}.

\begin{definition}
(\cite{aleksandrowicz2017computational}).
For any $k\in\NN$, let $\Dif{k}=\{L_1\cap L_2\colon L_1\in\Sig{k}\land L_2\in\Piclass{k}\}$.
\end{definition}

We note that \DifP{}=\Dif{1}.
Furthermore, for all $k\in\NN$, we have $\Sig{k}\cup\Piclass{k}\subseteq\Dif{k}$, since for any language $L$ we have $L=L\cap \Sigma^*$, and $\Sigma^*\in\Sig{k}\cap\Piclass{k}$. Additionally, we have that $\Dif{k}\subseteq\complexityclass{P^{\Sig{k}}}$, since two queries to a \Sig{k} oracle can be used to solve any problem in \Dif{k}. 

For an intuitive explanation why $\WDomStrat{}\in\Dif{2}$, one may observe---though this may not be immediately apparent---that \WDomStrat{} can alternatively be stated as follows: ``Is there a strategy for Player~1 that weakly dominates any other (without requiring strict domination for some opponent profile), but there do not exist two such strategies?'' This reformulation corresponds to the intersection of a \Sig{2} statement and a \Piclass{2} statement, thus placing the problem within \Dif{2}. 

\begin{wbox}
    \UQSAT\\ 
    \textbf{Input:} Two sets $\fXX=\{x_1, \dots,x_n\}$ and $\fYY=\{y_1, \dots,y_n\}$ of Boolean variables and a Boolean formula $\psi$ over $\fXX\cup \fYY$.\\
    \textbf{Question:}
    $\exists! \XX^*\in\bits{n}\;\forall\YY'\in\bits{n}\;\psi(\XX^*\concat \YY')=1$?
    Where $\exists!$ denotes that there is exactly one such $x^*$.
\end{wbox}

\begin{theorem}
\label{thm:wdom_uqsat}
$\UQSAT{}$ and $\WDomStrat{}$ are polynomial-time equivalent.
\end{theorem}

\begin{theorem}
\label{thm:uqsat_dif2}
$\UQSAT{}\in\Dif{2}$ and $\WDomStrat{}\in\Dif{2}$.
\end{theorem}

\begin{theorem}
\label{thm:uqsat_pi2}
\UQSAT{} and \WDomStrat{} are \Piclass{2}-hard.
\end{theorem}

\section{Discussion}
\label{sec:discussion}
Our work raises many open questions, some of which we highlight here. 
From a computational-complexity perspective, it would be intriguing to find unambiguous problems incomparable with \PTW{} and \PMA{} due to different combinatorial criteria for uniqueness. Such problems may shed light on additional natural subclasses of \USig{2}.

Furthermore, it would be valuable to establish tighter complexity-theoretic bounds for the proposed classes, e.g., by relating \PMA{} or \PCW{} to \MA{}, \BPP{}, or \LOPClass{}.
Notably, \MA{} and \BPP{} are unlikely to contain \PTW{}, \PCW{} or \PMA{}, since such inclusions would imply that \AM{} (Arthur-Merlin) contains \coNP{} (given that $\BPP{}\subseteq\MA{}\subseteq\AM{}$, see \cite{babai1985trading,russell1998symmetric}),
which in turn would collapse the polynomial hierarchy to its second level \cite{boppana1987does}.
The same argument shows that \PMA{} is unlikely to contain \PNP{} or even \NP{} (in sharp contrast to \PTW{} and \PCW{}).
Furthermore, it would be very interesting to find oracles separating the classes, besides the one separating \PNP{} and \PTW{} mentioned in \Cref{sec:introduction}.

Finally, we believe our framework may aid in classifying the complexity of unambiguous problems still open today. For instance, there are variations of hedonic games, such as fractional or modified fractional hedonic games, where the strong popularity question remains open \cite{BG24popularity,BrBu20a,Olse12a}. We conjecture that \PCW{} is the correct class for the mentioned instances.
\section{Acknowledgements}
Earlier versions of this work included a preliminary version of the results in \cite{gilboa2026StrongPopularity}, when the two works formed a single manuscript.
We thank Amnon Ta-Shma and Noam Ta-Shma for communicating to us their result of \Cref{thm:epsilon_pairwise_code_exists}.
We thank Rahul Santhanam for drawing our attention to the result on \LOPClass{} in \cite{korten2024strong}.
We thank Lance Fortnow for suggesting that \coAM{} could bound \PMA{}, and mentioning the work of \cite{FortnowYamakami96}.
Matan Gilboa was supported by the Engineering and Physical Sciences Research Council (EPSRC, grant EP/W524311/1).
Paul Goldberg was supported by the EPSRC (grants EP/X040461/1 and EP/X038548/1), and by the ARIA project “\textit{Aggregating Safety Preferences for AI Systems: A Social Choice Approach.}”
Noam Nisan was partially supported by a grant from the Israeli Science Foundation (ISF number 505/23).

\newpage
\section*{Appendix}
\addcontentsline{toc}{section}{Appendix}
In this appendix, we provide all missing proofs from the paper. \Cref{sec:PTW_Appendix} contains proofs corresponding to \Cref{sec:PTW}, \Cref{sec:PCW_Appendix} to \Cref{sec:PCW}, \Cref{sec:PMA_Appendix} to \Cref{sec:PMA}, \Cref{sec:PNP_Appendix} to \Cref{sec:PNP}, \Cref{sec:Sig2_Appendix} to \Cref{sec:Sig2}, and \Cref{sec:Beyond_sig2_Appendix} to \Cref{sec:Beyond_sig2}.

In case of nested proofs, we use $\triangle$ to mark the end of the inner one and $\square$ for the outer one.
We denote by $in_G(x):=|\{v\in V\colon (v,x)\in E\}|$ the in-degree of vertex $x\in V$ in $G$; analogously, $out_G(x):=|\{v\in V\colon (x,v)\in E\}|$.
Given $T\subseteq V$, we denote by $G[T]$ the subgraph induced by $T$ on $G$.
For $n\in\NN$, we use the notation $[n]:=\{1,\dots,n\}$ and $[n]_0:=\{0,\dots,n-1\}$.
We take $\log (n)$ to mean the base $2$ logarithm of $n$. 

Recall that given two bit strings $\XX$ and $\YY$ we write $\XX\YY$ or $\XX\concat\YY$ for their concatenation. 
We may combine both notations in the same expression, e.g., $\XX\YY\concat\XX'\YY'$, when it is helpful to emphasize that $\XX$ and $\YY$ (respectively $\XX'$ and $\YY'$) form natural pairs, even though the entire expression represents a single concatenation.
This notation also includes explicit bit strings or single bits, for instance $0\concat\XX$ or $10^n\XX$.
We often interpret bit strings as binary numbers, and use operations like $\XX<\YY$.
\appendix
\section{Proofs of Section \ref{sec:PTW} (\PTW{})}
\label{sec:PTW_Appendix}
In this section, we provide all proofs missing from \Cref{sec:PTW}. 
We begin with \Cref{thm:TournSource_PTW}.
\begin{proof}
(\textbf{\Cref{thm:TournSource_PTW}}).
We begin by arguing that $\TournSource$ can be reduced to $\WTournSource$. 
Given an instance $\CC$ of $\TournSource$, denote by $\CC'$ the circuit that outputs $1$ if
\[\big(\CC(\XX\concat \YY)=1 \land \CC(\YY\concat \XX)=0\big)\lor \big(\CC(\XX\concat \YY)=\CC(\YY\concat \XX) \land \XX>\YY\big)\]
and $0$ otherwise.
Thus, the question of $\TournSource$ becomes: 
\[\exists \XX\in\bits{n} \;\; \forall \YY\in\bits{n} \text{ with } \YY\neq \XX\;\; \CC'(\XX\concat \YY)=1\]
A careful review of the definition shows that $\CC'$ is anti-symmetric, in that $\CC'(\XX\concat \YY)=1-\CC'(\YY\concat \XX)$. Thus, this question is equivalent to
\[\exists \XX\in\bits{n} \;\; \forall \YY\in\bits{n} \text{ with } \YY\neq \XX\;\; \CC'(\XX\concat \YY)=1 \land \CC'(\YY\concat \XX)=0\]

which is exactly of the form of \WTournSource{}, and therefore, $\CC$ is a yes-instance of $\TournSource$ if and only if $\CC'$ is a yes-instance of \WTournSource{}.

We now wish to show $\WTournSource \reduce \TournSource$. We begin with an intuitive sketch of the reduction. Consider the graph $G$ induced by $\CC$, where $n$-bit strings are vertices, and a directed edge $(\XX,\YY)$ exists if and only if $\CC(\XX\concat \YY)=1\lor\CC(\YY\concat \XX)=0$. This is a weak tournament, as some pairs of vertices may have both edge between them (if $\CC(\XX\concat \YY)=\CC(\YY\concat \XX)$). We wish to resolve this issue by removing edges where needed, thus forming a tournament $G'$. We $G'$ to admit a solution if and only if $G$ does. Thus, we cannot simply remove edges arbitrarily, as we might accidentally create a winner. Specifically, we want to make sure that any vertex that tied with another vertex in $G$ remains a loser in $G'$. To do so, for each pair of vertices $\{\XX,\YY\}$ in $G$ we add a vertex $\mathtt{v}_{\XX,\YY}$. If $\XX$ and $\YY$ tied in $G$, we will make sure to create a cycle between $\XX$, $\YY$, and $\mathtt{v}_{\XX,\YY}$, thus eliminating ties while ensuring none of them is a winner $G'$. If in $G$ there was only one edge between $\XX$ and $\YY$, say $(\YY,\XX)$, then we add the edge $(\XX,\mathtt{v}_{\XX,\YY})$ to ensure $\mathtt{v}_{\XX,\YY}$ is not preventing $\XX$ from winning, if $\XX$ happens to be a solution in $G$.
Lastly, all added vertices by default lose to any of the original vertices of $G$ that are not associated with them.

We proceed to the formal construction. First, we introduce some definitions. 
Let $\XX<\YY$. We say the string $\XX 0^n$ is an \emph{original vertex} which \emph{corresponds} to $\XX$, and the string $\XX \YY$ is a \emph{structure vertex} which \emph{corresponds} to the unordered pair $\{\XX,\YY\}$; we also say $\XX \YY$ \emph{corresponds} to $\XX$ and to $\YY$ (individually). 
Original and structure vertices are \emph{valid} strings, while the rest are \emph{invalid}. 
That is, invalid strings are those of the form $\YY \XX$ (where $\XX<\YY$), if $\XX\neq 0^n$. 
In particular, we emphasize that if $\XX=0^n$, then the string $\YY 0^n$ corresponds to $\YY$, the string $0^{2n}$ corresponds to $\XX$, and the string $0^n \YY$ corresponds to the pair $\{\XX,\YY\}$.

Now, given a circuit $\CC$ with $2n$ inputs for $\WTournSource$, we construct a circuit $\CC'$ with $4n$ inputs. Given $\XX,\YY,\WW,\ZZ\in\bits{n}$, we define $\CC'(\XX\YY\concat \WW\ZZ)$ as follows.

First, we describe its behavior on the trivial cases, where non-corresponding or non-valid vertices are given as input.

\begin{itemize}
    \item If $\XX\YY=\WW\ZZ$, output $0$ (recall that we do not care about the behavior of the circuit on identical strings, as the formulation of the question only concerns non-identical strings).

    \item If both of the strings $\XX \YY$ and $\WW \ZZ$ are invalid, output $1$ if and only if $\XX \YY > \WW \ZZ$.
    
    \item If $\XX \YY$ is valid and $\WW \ZZ$ is invalid, output $1$, and if the converse holds output $0$.
    
    \item If both of the strings $\XX \YY$ and $\WW \ZZ$ are structure vertices, output $1$ if and only if $\XX \YY > \WW \ZZ$.
    
    \item If $\XX \YY$ is an original vertex and $\WW \ZZ$ is a structure vertex not corresponding to $\{\XX,\YY\}$, output $1$.
    
    \item If $\WW \ZZ$ is an original vertex and $\XX \YY$ is a structure vertex not corresponding to $\{\WW,\ZZ\}$, output $0$.
       
\end{itemize}

We now move on to the more interesting cases.
\begin{itemize}
    \item If both $\XX \YY$ and $\WW \ZZ$ are original vertices (i.e., $\YY=\ZZ=0^n$), calculate $c_1:=\CC(\XX\concat \WW)$ and $c_2:=\CC(\WW\concat \XX)$. If $c_1\neq c_2$, output $c_1$. If $c_1=c_2$, output $1$ if and only if $\XX>\WW$. Namely, we break the tie somewhat arbitrarily using lexicographic ordering (we avoid unintentionally creating a winner in the following cases).
    
    \item Suppose $\XX \YY$ is an original vertex and $\WW \ZZ$ is a structure vertex corresponding to $\XX$.
    Then by definition $\ZZ>\WW$, and either $\XX=\ZZ$ or $\XX=\WW$.
    First, assume $\XX=\ZZ$. Calculate $c_1:=\CC(\XX\concat \WW)$ and $c_2:=\CC(\WW\concat \XX)$. If $c_1>c_2$, output $1$. If $c_1\leq c_2$, output $0$ (the case $c_1=c_2$ is included here because we have $\XX>\WW$). Thus, we have that $\XX \YY$ (i.e. $\XX 0^n$) wins against $\WW 0^n$, so we maintain that neither of the strings is a winner.
    Now, assume $\XX=\WW$. Then, calculate $c_1:=\CC(\XX\concat \ZZ)$ and $c_2:=\CC(\ZZ\concat \XX)$. Symmetrically to the previous case, if $c_1\geq c_2$, output $1$. If $c_1<c_2$, output $0$ (in this case the tie breaker is in favor of $\XX 0^n$, since $\XX 0^n$ already loses to $\ZZ 0^n$).

    \item The case where $\WW \ZZ$ is an original vertex and $\XX \YY$ is a structure vertex corresponding to $\WW$ is defined anti-symmetrically (one may define the output by $\CC'(\XX \YY\concat \WW \ZZ)=1-\CC'(\WW \ZZ\concat \XX \YY)$, since $\CC'(\WW \ZZ\concat \XX \YY)$ is well-defined above).
\end{itemize}

We argue that there exists a $\WTournSource$ solution in $\CC$ if and only if there exists a $\TournSource$ solution in $\CC'$.
Suppose there exists a $\WTournSource$ solution $\XX^*$ in $\CC$, and consider $\XX^* 0^n$. It is immediate that, in $\CC^*$, $\XX^* 0^n$ beats any invalid vertex or non-corresponding structure vertex. Additionally, since $\XX^*$ is a winner in $\CC$, $\XX^*0^n$ must beat any original vertex in $\CC'$. 
It then follows that for any corresponding structure vertex $\mathtt{v}_{\XX^*,\YY}$ for some $\YY$, vertex $\XX^*$ beats $\YY$, and thus by design $\mathtt{v}_{\XX^*,\YY}$ loses to $\XX^*$.
Similarly, for any corresponding structure vertex $\mathtt{v}_{\YY,\XX^*}$ for some $\YY$, vertex $\XX^*$ beats $\YY$, and thus by design $\mathtt{v}_{\YY,\XX^*}$ loses to $\XX^*$.

In the other direction, if $\XX^*=\XX_1 \XX_2$ is a solution to $\TournSource$ in $\CC'$, where $\XX_1,\XX_2\in\bits{n}$, then first we must have that $\XX_2=0^n$, as otherwise $\XX^*$ is other invalid or a structure vertex, in which case it loses to any (non-corresponding) original vertex. Secondly, we must have that $\XX_1$ is a solution to $\WTournSource{}$ in $\CC$, since if it loses to some $\XX'$ in $\CC$ then $\XX^*$ loses to $\XX' 0^n$ in $\CC'$, a contradiction; and if $\XX_1$ ties with some $\XX'$ in $\CC$, then by design $\XX^*$ must lose either to $\XX' 0^n$ or to $\mathtt{v}_{\XX_1,\XX'}$ (depending on whether $\XX^*>\XX'$ or not). Again, we have a contradiction to $\XX^*$ being a solution to $\CC'$.
\end{proof}

To prove \Cref{thm:PTW_S2P}, we begin by establishing the following lemma, which (formulated slightly differently) is attributed by \cite{reid2004domination} (Theorem 4) to Erd{\H{o}}s.
It asserts that if there is no winner in a tournament then we can find a short certificate for it, namely a set $S$ of vertices of logarithmic size such that all vertices lose to some $s\in S$.

\begin{lemma}
\label{lem:exists_strong_set}
(Erd{\H{o}}s) Let $G=(V,E)$ be a tournament. If there is no source in $G$, then there exists a set of vertices $S\subseteq V$ with $|S|\leq \lceil\log(|V|)\rceil+1$ such that for all $v\in V$ (including $v\in S$) we have $(s,v)\in E$ for some $s\in S$.
\end{lemma}

\begin{proof}
First, we prove the following proposition:
\begin{prop}
\label{prop:exists_strong_vertex}
Let $G=(V,E)$ be a tournament with $|V|\geq 2$.
Then there exists a vertex $v^*$ with $out_G(v^*)\geq \big\lceil\frac{|V|-1}{2}\big\rceil$, namely:
\[|\{v\in V\colon (v^*,v)\in E\}|\geq \big\lceil\frac{|V|-1}{2}\big\rceil\]
\end{prop}

\begin{proof}
Since $G$ is a tournament, there are exactly $\binom{|V|}{2} = \tfrac{|V|(|V|-1)}{2}$ directed edges in the graph. 
Thus, the average out-degree of the vertices is $\tfrac{|V|(|V|-1)}{2} \,/\, |V| = \tfrac{|V|-1}{2}$, 
implying there must exist a vertex with out-degree at least $\tfrac{|V|-1}{2}$. 
Since the out-degree of a vertex is an integer, we may round this up to obtain the required result.
\subqed
\end{proof}

Now, let us consider the following algorithm.
\begin{algorithm}
    \caption{Construction of logarithmic winning set}
    \label{alg:log_winning_set}
    \begin{algorithmic}[1]
        \State \textbf{Input:} Tournament $G=(V,E)$.
        \State \textbf{Output:} Set $S\subseteq V$ that beats any vertex.
        \State $S \gets \emptyset$
        \State $T \gets V$\label{alg_step:iteration_start}
        \State Pick $x\in argmax\{out_{G[T]}(x)\colon x\in T\}$\label{alg_step:pick_best_x}
        \State $S\gets S\cup\{x\}$
        \State $T\gets \{v\in V\colon \forall \; s\in S \text{ with } s\neq v \text{ we have } (v,s)\in E\}$ // vertices who do not lose to any $s\in S$.
        \State If $|T|\geq 2$, go back to \Cref{alg_step:pick_best_x}. If $|T|=1$ with $T=\{t\}$ for some vertex $t\in V$, pick $x\in V\bs S$ such that $(x,t)\in E$, set $S\gets S\cup\{x\}$, and return $S$.\label{alg_step:end}
    \end{algorithmic}
\end{algorithm}

We wish to show the output set $S$ of the algorithm satisfies the conditions of the lemma.
In each iteration $T$ is defined as the set of uncovered vertices in $V$, namely those that do not lose to any $s\in S$. Since the algorithm only stops when $T=\emptyset$, we have that upon termination, $S$ indeed covers the entire set $V$, as all vertices in $V$ lose to some $s\in S$.

It remains to show that the algorithm terminates and that $|S|$ is sufficiently small.
Denote by $T_i$ the state of the set $T$ from \Cref{alg:log_winning_set} in iteration $i$.
By \Cref{prop:exists_strong_vertex} (applied on $G[T]$), we have that $|T_{i+1}|\le |T_i|-\lceil\frac{|T_i|-1}{2}\rceil\le \big\lceil \frac{|T_i|}{2}\big\rceil$ for all iterations $i$ except the last (when $|T|=1$); a simple case analysis on the parity of $|T_i|$ shows why this holds.
Furthermore, \Cref{alg_step:end} is well defined, since there is no source in $G$, and therefore there must exist a vertex $x$ that covers the final remaining vertex $t$. The first observation implies that the number of iterations before we reach $|T|=1$ is bounded by $\lceil\log(|V|)\rceil$, while the second observation establishes that the algorithm indeed terminates.
Since in each iteration $|S|$ increases by one, we conclude that at the end of the run $|S|\leq \lceil\log(|V|)\rceil+1$.
\end{proof}

We can now prove \Cref{thm:PTW_S2P}.
\begin{proof}
(\textbf{\Cref{thm:PTW_S2P}}).
To prove this, we show that $\TournSource \in \SymP$.
Let $\CC$ be an instance of $\TournSource$. It is useful to rephrase the problem of this instance as follows. For all $\XX,\YY\in\bits{n}$ define the circuit $\CC'$ by:
\[
\CC'(\XX\concat \YY)=
\begin{cases}
    1 & \text{if } \YY\neq\XX \land \Big(\big(\CC(\XX\concat \YY)=1 \land \CC(\YY\concat \XX)=0\big)\lor \big(\CC(\XX\concat \YY)=\CC(\YY\concat \XX) \land \XX>\YY\big)\Big)\\
    0 & otherwise
\end{cases}
\]
Thus, the question becomes:
\[\exists \XX\; \forall \YY\neq\XX\;\;\CC'(\XX\concat \YY)=1\]
It follows from the definitions that the answer to this question is yes if and only if $\CC$ is a yes-instance of $\TournSource$.
This representation is convenient as $\CC'$ induces a tournament on $\bits{n}$, since for all $\XX,\YY\in\bits{n}$ with $\XX\neq\YY$ we have $(\CC'(\XX\concat \YY)=1\land \CC'(\YY\concat \XX)=0)\lor(\CC'(\XX\concat \YY)=0\land \CC'(\YY\concat \XX)=1)$.

Given a set of strings $S\subseteq \bits{n}$ with $|S|\leq n+1$, and a string $v\in\bits{n}$, let $P(\CC',S,v)$ be the predicate that outputs $1$ if and only if $\CC'(s\concat v)=0$ for all $s\in S$ (that is, the vertex $v$ loses to no vertex in $S$). Since $|S|\leq n+1$, $P$ can be computed in polynomial time.

Now, if there exists a source $v^*$ in the graph induced by $\CC'$, then for all $S\subseteq\bits{n}$ we have $P(\CC',S,v^*)=1$, since $v^*$ loses to no vertex (if $v^*$ itself is in $S$ we still have $P(\CC',S,v^*)=1$, since $\CC'(v^*\concat v^*)=0$).

On the other hand, if there does not exist such a source, then by \Cref{lem:exists_strong_set} there exists a set $S\subseteq\bits{n}$ with $|S|\leq n+1$ such that for all $v\in\bits{n}$ we have that $v$ loses to some $s\in S$, namely $P(\CC',S,v)=0$.

Thus, by definition we have that $\TournSource\in\SymP$.
\end{proof}

\begin{proof}
(\textbf{\Cref{thm:L2P_PTW}}).
Given a language $L\in\LOPClass{}$, let $P(\cdot,\cdot,\cdot)$ be a polynomial-time predicate and $p$ a polynomial, such that for all $w\in\Sigma^*$, $P(w,\cdot,\cdot)$ defines a total order on $\bits{p(n)}$ whose minimal element $a$ has $a_1=L(w)$.
Let $w\in \Sigma^*$ with $|w|=n$.
We construct a Boolean circuit $\CC_{w}:\bits{2p(n)}\to\{0,1\}$, where for all $\XX,\YY\in\bits{p(n)}$ we let $\CC_{w}(\XX,\YY)=1$ if and only if $P(w,\XX,\YY)=1\wedge \XX_1=1$ (where $\XX_1$ denotes the first bit of $\XX$).
If $w\in L$, then the minimal element $\XX^*$ in the linear order $P(w,\cdot,\cdot)$ is clearly a solution for \WTournSource{} on $\CC_w$. Otherwise, there cannot be a solution: The minimal element $\XX^*$ does not satisfy the condition on its first bit, and any other string must lose to someone on account of not being minimal in the linear order.
\end{proof}

\begin{proof}
(\textbf{\Cref{thm:coNP_source_PTW}}).
It is immediate that $\WTournSource$ is polynomial-time reducible to $\MultiTournSource$, by ignoring the last $n$ input bits. 

We now show that $\MultiTournSource\reduce\WTournSource$. Given an instance $\CC$ of $\MultiTournSource$ with $3n$ input bits, we construct a circuit $\CC'$ with $4n$ input bits. On inputs $\XX_1,\XX_2,\YY,\ZZ\in\bits{n}$ we define:
\[
\CC'(\XX_1\concat \XX_2\concat \YY\concat \ZZ)=
\begin{cases}
    1 & \text{if } \XX_2=0^n \land \big(\XX_1\neq\YY\implies(\CC(\XX_1\concat \YY\concat \ZZ)=1 \land \CC(\YY\concat \XX_1\concat \ZZ)=0)\big)\\
    0 & otherwise
\end{cases}
\]
We prove that there exists a solution in the original instance if and only if there exists one in the reduced instance.
Assume that there exists $\XX^*\in\bits{n}$ such that $\forall \YY,\ZZ\in\bits{n}$ with $\YY\neq \XX^*$ we have that $\CC(\XX^*\concat \YY\concat \ZZ)=1\land \CC(\YY\concat \XX^*\concat \ZZ)=0$. 
Consider the string $\XX^*0^n$. Let $\YY,\ZZ\in\bits{n}$ such that $\YY\ZZ\neq\XX^* 0^n$.
If $\YY=\XX^*$ then we must have $\ZZ\neq 0^n$. 
Thus, by definition of $\CC'$ we have that $\CC'(\XX^*0^n\concat \YY\ZZ)=1$ (since $\XX^*=\YY$), while $\CC'(\YY\ZZ\concat \XX^*0^n)=0$ (since $\ZZ\neq 0^n$).
If $\YY\ne \XX^*$, then the choice of $\XX^*$ ensures that $\CC'(\XX^*0^n\concat \YY\ZZ)=1\land \CC'(\YY\ZZ\concat \XX^*0^n)=0$.
Hence, $\XX^*0^n$ is a solution for $\WTournSource{}$ on $\CC'$. 

In the opposite direction, assume there exist $\XX_1,\XX_2\in\bits{n}$ such that $\forall \YY,\ZZ\in\bits{n}$ with $\YY\ZZ\neq\XX_1\XX_2$ we have $\CC'(\XX_1\XX_2\concat \YY\ZZ)=1\land \CC'(\YY\ZZ\concat \XX_1\XX_2)=0$. 
Let $\YY,\ZZ\in\bits{n}$ with $\YY\neq\XX_1$. In particular, this implies $\YY\ZZ\neq\XX_1\XX_2$, and therefore $\CC'(\XX_1\XX_2\concat \YY\ZZ)=1$. 
Thus, by definition of $\CC'$, and since $\XX_1\neq\YY$, we have that $\CC(\XX_1\concat \YY\concat \ZZ)=1\land \CC(\YY\concat \XX_1\concat \ZZ)=0$. 
Hence, $\XX_1$ is a solution for $\MultiTournSource$ on $\CC$.
\end{proof}
\section{Proofs of Section \ref{sec:PCW} (\PCW{})}
\label{sec:PCW_Appendix}
In this section, we provide all proofs missing from \Cref{sec:PCW}.

Throughout, given an instance $\langle \CC_1,\dots,\CC_k\rangle$ of \CKTCondorcet{} or $\kCKTCondorcet{k}$, we say a bit string $\XX$ is \emph{more popular} than $\YY$ if $\XX$ beats $\YY$ in a majority vote among the circuits, namely if $\sum_{i=1}^k\big(\sgn(\CC_i(\XX)-\CC_i(\YY))\big)>0$.

\begin{proof}
(\textbf{\Cref{thm:PCW_PTW}}).
\CKTCondorcet{} can be reduced to \WTournSource{} by constructing a circuit which, given two strings, outputs $1$ if the first is more popular than the second in the original instance, and $0$ otherwise.
A string $x\in\bits{n}$ is thus a Condorcet winner in the original instance if and only if it is a source in the reduced one.
\end{proof}

\begin{proof}
(\textbf{\Cref{thm:Condorcet_parameterized}}).
\Cref{thm:Condorcet_parameterized:kk'} holds since we can pad an instance of $\kCKTCondorcet{k}$ with $k'-k$ degenerate circuits that always output 0, thus not affecting the result of a majority vote of any pair of strings.
\Cref{thm:Condorcet_parameterized:PCW} holds trivially since an instance of $\kCKTCondorcet{k}$ is by definition an instance of \CKTCondorcet{}.
\Cref{thm:Condorcet_parameterized:PNP} holds for $k=1$ as this yields exactly the problem \UOPT{}, which by \Cref{lem:UOPT} is \PNP{}-complete; for any $k\ge 2$, the result then follows from \Cref{thm:Condorcet_parameterized:kk'}.
\end{proof}

\begin{proof}
(\textbf{\Cref{thm:PNP_PCW}}).
The proof follows from \Cref{thm:Condorcet_parameterized:PCW,thm:Condorcet_parameterized:PNP} of \Cref{thm:Condorcet_parameterized}.
\end{proof}

We now turn our attention to proving \Cref{thm:all_dice_PCW_complete}.
We begin with establishing \PCW{}-hardness of the considered problems, followed by containment in \PCW{}.

\begin{lemma}
\label{lem:gdice_cktdice}
$\GDice{}\reduce\CKTDice{}$.
\end{lemma}

\begin{proof}
Given an edge-weighted graph we can construct a circuit for each vertex which gets an encoding of a partition as input and outputs the sum of outgoing edge-weights of its vertex in the given partition. The calculations are then identical in both games, and thus there is a winner in the former if and only if there is one in the latter.
\end{proof}

We now wish to show that \GDice{} is \PCW{}-hard.
To do so, we reduce from \ASHG{}, which is \PCW{}-complete by \Cref{thm:ashg_pop_pcw}.
We begin with a formal description of the problem \ASHG{}.
Let $N$ be a set of agents. 
A \emph{coalition} is a non-empty subset of $N$.
For an agent $i\in N$, we denote by $\pi(i)$ the coalition $i$ belongs to in $\pi$. 
An ASHG is specified by a pair $\langle N,\mathbf{v} \rangle$ where $N$ is a set of agents, and $\mathbf{v} = (\mathbf{v}_i\colon N\to \RR)_{i\in N}$ is a collection of \textit{valuation functions}. 
The quantity $\mathbf{v}_i(j)$ denotes the value agent $i$ assigns to agent $j$. 
We define the utility of agent $i$ in coalition $S$ by $u_i(S)=\sum_{j\in S\bs\{i\}}\mathbf{v}_i(j)$, namely the sum of values she assigns to the other members of her coalition. 
We extend the utilities for partitions by setting $u_i(\pi)=u_i(\pi(i))$, for any agent $i$ and partition $\pi$.
We say partition $\pi$ is \emph{more popular} than $\pi'$ if $\sum_{i=1}^N \sgn (u_i(\pi)-u_i(\pi'))>0$; That is, the number of agents who prefer $\pi$ is greater than the number of agents who prefer $\pi'$ (indifferent agents are not counted).
We say $\pi$ is \emph{strongly popular} if $\pi$ is more popular than any partition $\pi'\ne\pi$.

\begin{wbox}
\label{def:ASHG}
\ASHG{}\\
\textbf{Input:} Additively separable hedonic game $\langle N,\mathbf{v} \rangle$.\\
\textbf{Question:} Does $\langle N,\mathbf{v} \rangle$ admit a strongly popular partition?
\end{wbox}

\begin{lemma}
\label{lem:gdice_pcw}
\GDice{} is \PCW{}-hard.
\end{lemma}
\begin{proof}
We reduce from \ASHG{}. 
The idea of the reduction is to modify the weights such that the ranges of values that different vertices may obtain are pairwise disjoint. 
Thus, comparisons between different vertices becomes irrelevant since one vertex dominates the other regardless of the input dice, and so a strongly popular partition in the original instance is equivalent to a winning die in the reduced instance. 
Formally, given an instance $\langle N,v\rangle$ of \ASHG{}, let 
\begin{equation}
U = \sum_{\substack{i,j \in [n] \\ v_i(j) > 0}} v_i(j),\text { and }\;L = \sum_{\substack{i,j \in [n] \\ v_i(j) < 0}} |v_i(j)|
\end{equation}
be the sums of all positive and absolute values of all negative values that agents assign to each other in the given hedonic game, respectively. Define the directed graph $G=(V,w)$ with $V=N$ and for all $i,j\in N$, let:
\[
w_i(j)=
\begin{cases}
    v_i(j) & \text{if } i\neq j\\
    i\cdot (U+L+1) & \text{if } i=j
\end{cases}
\]
Notice that weights correspond to valuations of the original instance, except for self-loops which are modified proportionally to the vertex number (whereas in ASHG agents assign value $0$ to themselves). Therefore, with $s$ as in the definition of \GDice{}, for any partition $\pi$ we have that $s_i(\pi)\in[i\cdot (U+L+1)-L,i\cdot (U+L+1)+U]$, which can also be written as 
\[s_i(\pi)\in[i\cdot (U+L+1)-L,(i+1)\cdot (U+L+1)-L-1].\]
Thus, given two partitions $\pi_1$ and $\pi_2$ we have $s_i(\pi_1)\leq (i+1)\cdot (U+L+1)-L-1<(i+1)\cdot (U+L+1)-L\leq s_{i+1}(\pi_2)$.
Inductively, for vertices $i<j$ we have that 
\begin{equation}
\label{eq_pop_dice:ranges_by_index}
s_i(\pi_1)<s_j(\pi_2).
\end{equation} 
Furthermore, notice that 
\begin{equation}
\label{eq_pop_dice:order_preserved}
s_i(\pi_1)>s_i(\pi_2) \iff u_i(\pi_1)>u_i(\pi_2)
\end{equation} 
(namely, the preference order of each individual agent over the partitions is preserved), since $s_i(\pi_1)=u_i(\pi_1)+i\cdot (U+L+1)$ and $s_i(\pi_2)=u_i(\pi_2)+i\cdot (U+L+1)$. Thus, for any two partitions $\pi_1$ and $\pi_2$ we have:
\[
\sum_{i=1}^n\sum_{j=1}^n\big(\sgn(s_i(\pi_1)-s_j(\pi_2))\big)=\]
\[
\sum_{i=1}^n
\Bigg(\sum_{j=1}^{i-1}\big(\sgn(s_i(\pi_1)-s_j(\pi_2))\big) + 
\sum_{j=i+1}^n\big(\sgn(s_i(\pi_1)-s_j(\pi_2)) \big)
+
\big(\sgn(s_i(\pi_1)-s_i(\pi_2))\big) \Bigg)=
\]
\[
\sum_{i=1}^n
\Bigg((i-1) - (n-i) +
\big(\sgn(s_i(\pi_1)-s_i(\pi_2))\big) \Bigg)=
\sum_{i=1}^n(2i-n-1) + \sum_{i=1}^n\big(\sgn(u_i(\pi_1)-u_i(\pi_2))\big)=
\]
\[\sum_{i=1}^n\big(\sgn(u_i(\pi_1)-u_i(\pi_2))\big)\]

where the second equality stems from \Cref{eq_pop_dice:ranges_by_index}, the third from \Cref{eq_pop_dice:order_preserved}, and the fourth is due to:
\[
\sum_{i=1}^n(2i-n-1)=2\sum_{i=1}^n{i}-\sum_{i=1}^n{n}-\sum_{i=1}^n{1}=n(n+1)-n^2-n=0.
\]
Therefore, we have
$\sum_{i=1}^n\sum_{j=1}^n\big(\sgn(s_i(\pi_1)-s_j(\pi_2) \big)>0$
if and only if 
$\sum_{i=1}^n\big(\sgn(u_i(\pi_1)-u_i(\pi_2))\big)>0$,
and so $\pi_1$ is a solution for \GDice{} on $G$ if and only if it is a solution for \ASHG{} on $\langle N,v\rangle$.
\end{proof}

\begin{lemma}
\label{lem:cktdice_pcw}
\CKTDice{} is \PCW{}-hard.
\end{lemma}
\begin{proof}
This follows immediately from \Cref{lem:gdice_cktdice,lem:gdice_pcw}.
\end{proof}

So far, we have seen that both formulations of the dice problem are \PCW{}-hard. 
We now wish to reduce from \CKTDice{} to \CKTCondorcet{} to show it is in \PCW{}. 
We begin with an intuitive explanation. 
Consider the following randomized reduction (which we will de-randomize). 
Given an instance $\fCC=\langle\CC_1,\dots,\CC_m\rangle$ of \CKTDice{}, we create a larger collection of $M>m$ circuits $\fCC'=\langle\CC'_1,...,\CC'_M\rangle$, such that for $i\in[M]$ we set $\CC'_i(\XX)=\CC_j(\XX)$ where $j$ is chosen from $[m]$ uniformly at random.
Let $\XX$ and $\YY$ be two dice, and consider two faces $l,k\in[m]$. 
For how many values of $i\in[M]$ do we have that $\CC'_i(\XX)=\CC_l(\XX)$, and $\CC'_i(\YY)=\CC_k(\YY)$? 
If $M$ is large enough, with high probability we have that the number of such values will be almost identical for all $(l,k)$ pairs, and therefore we have that if die $\XX$ beats die $\YY$ in $\fCC$, then with high probability $\XX$ is more popular than $\YY$ in $\fCC'$. 

However, we encounter a problem with this idea: 
If two dice tie with each other, the randomized process will most likely result in one of them being more popular than the other in the reduced Condorcet instance, and thus unintentional Condorcet winners may be created. 
To cope with this issue, we introduce a \emph{strict} version of \CKTDice{}, namely \StrictCKTDice{}, where we break ties using lexicographic order of the dice labels. 
We show that the original version can be reduced to the strict one, and then proceed with formalizing a de-randomized version of the construction described above.

The de-randomization process relies on a result from coding theory by Ta-Shma and Ta-Shma (personal communication), who show that there exists a code with a pseudo-randomness property that ensures the desired functionality of the reduction, as described above. 
That is, instead of setting $\CC'_i(\XX)=\CC_j(\XX)$ with $j$ chosen randomly, we choose $j$ deterministically according to this code.
The notation in \Cref{thm:all_dice_PCW_complete} differs from the rest of the section in order to correspond with coding theory literature, as appears in \Cref{def:epsilon_pairwise_code,def:poly_constructible_set} and \Cref{thm:epsilon_pairwise_code_exists}.

\begin{wbox}
    \StrictCKTDice{}\\ 
    \textbf{Input:} Sequence of Boolean circuits $\fCC=\langle\CC_1,\dots,\CC_m\rangle$, with $\CC_i\colon\bits{n}\to\bits{n'}$ for all $i\in[m]$.\\
    \textbf{Question:} Does the following hold:
    \[\exists \XX\in\bits{n} \;\; \forall \YY\in\bits{n} \text{ with } \YY\neq \XX\]
    \[\Bigg(\sum_{i=1}^m\sum_{j=1}^m\big(\sgn(\CC_i(\XX)-\CC_j(\YY))\big)>0\Bigg) \lor \Bigg(\sum_{i=1}^m\sum_{j=1}^m\big(\sgn(\CC_i(\XX)-\CC_j(\YY))\big)=0 \land \XX>\YY \Bigg)\]
\end{wbox}

\begin{lemma}
\label{lem:dice_reduce_strict_dice}
$\CKTDice{}\reduce\StrictCKTDice{}$.
\end{lemma}
\begin{proof}
Suppose we are given an instance $\fCC=\langle\CC_1,\dots,\CC_m\rangle$ of \CKTDice{} with $\CC_i\colon\bits{n}\to\bits{n'}$ for all $i\in[m]$.
We construct an instance $\fCC'=\langle\CC'_1,\dots, \CC'_{4m+3}\rangle$ of \StrictCKTDice{}, where $\CC_i\colon\bits{2n+2}\to\bits{n'+1}$ for all $i\in[4m+3]$. To describe the circuits' behavior, we introduce some terminology.
We say a vertex (string) $\ZZ\in\bits{2n+2}$ is an \emph{original} vertex that \emph{corresponds} to vertex $\XX\in\bits{n}$ if $\ZZ=0^{n+2}\XX$. 
We say $\ZZ$ is an \emph{edge-vertex} that \emph{corresponds} to the vertices $\XX$ and $\YY$, and to the edge $(\XX,\YY)$, if $\ZZ=j k \XX \YY$, where $j,k\in\{0,1\}$, and $(j,k)\neq (0,0)$ (that is, $\ZZ$ starts with $01$, $10$, or $11$). Note that there are three edge-vertices for any edge.
We say $\ZZ$ is \emph{valid} if it is an original vertex or an edge-vertex, and \emph{invalid} otherwise.

We now describe the behavior of the circuits in $\fCC'$. Let $\ZZ\in\bits{2n+2}$.
If $\ZZ$ is invalid, we let $\CC'_i(\ZZ)=0$ for all $i\in[4m+3]$.
If $\ZZ$ is valid, we distinguish between the first $4m$ circuits and the last three. The first $4m$ circuits are intuitively meant to multiply the original circuits, to enhance any advantage a vertex had over another vertex in the original game. Thus, let $i\in[m]$. 
If $\ZZ$ is an original vertex corresponding to $\XX\in\bits{n}$, we set $\CC'_{4i-j}(\ZZ):=\CC_i(\XX)$ for all $j\in\{0,1,2,3\}$.
If $\ZZ$ is an edge-vertex corresponding to $(\XX,\YY)\in\bits{n}\times\bits{n}$, let $DL_{\XX,\YY}\in\{\XX,\YY\}$ be the Dice-loser between $\XX$ and $\YY$ (if there is a tie, we choose $\XX$), that is 
\[DL_{\XX,\YY}:=
\begin{cases}
    \XX & \text{ if } \sum_{i=1}^m\sum_{j=1}^m\big(\sgn(\CC_i(\XX)-\CC_j(\YY))\big)\leq 0\\
    \YY & \text{ otherwise}
\end{cases}
\]
Then, we set $\CC'_{4i-j}(\ZZ)=\CC_i(DL_{\XX,\YY})$ for all $j\in\{0,1,2,3\}$.

We proceed with describing the last three circuits, $\CC'_{4m+1}$, $\CC'_{4m+2}$, and $\CC_{4m+3}$. 
If $\ZZ$ is an original vertex, we set $\CC'_{4m+i}(\ZZ)=2^{n'}$ for each $i\in[3]$.
If $\ZZ$ is an edge-vertex corresponding to $(\XX,\YY)\in\bits{n\times\bits{n}}$, we distinguish further according to the prefix of $\ZZ$.
If $\ZZ$ begins with $01$, we set 
$\CC'_{4m+1}(\ZZ)=2^{n'}+2$, $\CC'_{4m+2}(\ZZ)=2^{n'}+4$, and $\CC'_{4m+3}(\ZZ)=2^{n'}+9$.
If $\ZZ$ begins with $10$, we set 
$\CC'_{4m+1}(\ZZ)=2^{n'}+1$, $\CC'_{4m+2}(\ZZ)=2^{n'}+6$, and $\CC'_{4m+3}(\ZZ)=2^{n'}+8$.
If $\ZZ$ begins with $11$, we set 
$\CC'_{4m+1}(\ZZ)=2^{n'}+3$, $\CC'_{4m+2}(\ZZ)=2^{n'}+5$, and $\CC'_{4m+3}(\ZZ)=2^{n'}+7$.

The intuition is that three edge-vertices $v_1$, $v_2$, and $v_3$ corresponding to the same edge form a ``winning cycle'' in the dice game, since the values they obtain from circuits $\CC'_{4m+1}$, $\CC'_{4m+2}$, and $\CC'_{4m+3}$ simulate an intransitive dice instance. 
This ensures $v_1$, $v_2$, and $v_3$ cannot be a solution for \StrictCKTDice{} on $\fCC'$; thus, we avoid creating an unintentional winner in the reduction.

Formally, we observe that, on valid vertices, circuits $\CC'_1,\dots,\CC'_{4m}$ output values upper bounded by $2^{n'}-1$, while circuits $\CC'_{4m+1}$, $\CC'_{4m+2}$, and $\CC'_{4m+3}$ output values lower bounded by $2^{n'}$. Thus, for all valid vertices $\ZZ,\ZZ'\in\bits{2n+2}$ we have 
\begin{equation}
\label{eq:dice_reduce_strict_dice:sum_separation}
\sum_{i=1}^{4m+3}\sum_{i'=1}^{4m+3}\big(\sgn(\CC'_i(\ZZ)-\CC'_{i'}(\ZZ'))\big)=
\end{equation}
\[
\sum_{i=1}^{4m}\sum_{{i'}=1}^{4m}\big(\sgn(\CC'_i(\ZZ)-\CC'_{i'}(\ZZ'))\big)+
\sum_{i=4m+1}^{4m+3}\sum_{i'=4m+1}^{4m+3}\big(\sgn(\CC'_i(\ZZ)-\CC'_{i'}(\ZZ'))\big)
\]
as all other terms are cancel-out (e.g., $\sgn(\CC'_{4m+1}(\ZZ)-\CC'_{m}(\ZZ'))+\sgn(\CC'_{m}(\ZZ)-\CC'_{4m+1}(\ZZ'))=1-1=0$).
Furthermore, if $\ZZ$ and $\ZZ'$ are original vertices corresponding to $\XX$ and $\XX'$ respectively, we have:
\begin{equation}
\label{eq:sum_16_original_original}
\sum_{i=1}^{4m}\sum_{i'=1}^{4m}\big(\sgn(\CC'_i(\ZZ)-\CC'_{i'}(\ZZ'))\big)=
\sum_{i=1}^{m}\sum_{j=0}^{3}\sum_{i'=1}^{m}\sum_{j'=0}^{3}\big(\sgn(\CC'_{4i-j}(\ZZ)-\CC'_{4i'-j'}(\ZZ'))\big)=
\end{equation}
\[\sum_{i=1}^{m}\sum_{j=0}^{3}\sum_{i'=1}^{m}\sum_{j'=0}^{3}\big(\sgn(\CC_i(\XX)-\CC_{i'}(\XX'))\big)=
16\cdot\sum_{i=1}^{m}\sum_{i'=1}^{m}\big(\sgn(\CC_i(\XX)-\CC_{i'}(\XX'))\big).
\]

If $\ZZ$ is an original vertex corresponding to $\XX$ and $\ZZ_{\XX',\XX''}$ is an edge-vertex, we have:
\begin{equation}
\label{eq:sum_16_edge_original}
\sum_{i=1}^{4m}\sum_{i'=1}^{4m}\big(\sgn(\CC'_i(\ZZ)-\CC'_{i'}(\ZZ_{\XX',\XX''}))\big)=
\sum_{i=1}^{m}\sum_{j=0}^{3}\sum_{i'=1}^{m}\sum_{j'=0}^{3}\big(\sgn(\CC'_{4i-j}(\ZZ)-\CC'_{4i'-j'}(\ZZ_{\XX',\XX''}))\big)=
\end{equation}
\[
16\cdot\sum_{i=1}^{m}\sum_{i'=1}^{m}\big(\sgn(\CC_i(\XX)-\CC_{i'}(DL_{\XX',\XX''}))\big)\]

We now establish several claims.
\begin{prop}
\label{prop:dice_reduce_strict_dice:valid_vs_invalid}
    Let $\ZZ,\ZZ'\in\bits{2n+2}$, where $\ZZ$ is a valid vertex and $\ZZ'$ is invalid. Then $\ZZ$ beats $\ZZ'$ on $\fCC'$.
\end{prop}
\begin{proof}
    We have that $\CC'_i(\ZZ')=0$ for all $i\in[4m+3]$, while $\CC'_i(\ZZ)\ge 0$ for all $i\in[4m+3]$ and $\CC'_{4m+1}(\ZZ)\geq 2^{n'}>0$. 
    Hence, the number of pairs $(i,i')\in [4m+3]$ such that $\CC'_i(\ZZ)>\CC'_{i'}(\ZZ')$ is at least $4m+3$, while the number of pairs $(i,i')\in [4m+3]$ such that $\CC'_i(\ZZ)<\CC'_{i'}(\ZZ')$ is $0$.
    Therefore $\ZZ$ beats $\ZZ'$.
\subqed
\end{proof}
\begin{prop}
\label{prop:dice_reduce_strict_dice:original_vs_original}
    Let $\XX,\XX'\in\bits{n}$, and denote by $\ZZ=0^{n+2}\XX$ and $\ZZ'=0^{n+2}\XX'$ the corresponding original vertices.
    If $\XX$ beats $\XX'$ on $\fCC$, then $\ZZ$ beats $\ZZ'$ on $\fCC'$.
\end{prop}
\begin{proof}
Since $\XX$ beats $\XX'$ on $\fCC$, we have $\sum_{i=1}^{m}\sum_{i'=1}^{m}\big(\sgn(\CC_i(\XX)-\CC_{i'}(\XX'))\big)\geq 1$.  Thus, by \Cref{eq:dice_reduce_strict_dice:sum_separation,eq:sum_16_original_original} we have that
$\sum_{i=1}^{4m+3}\sum_{i'=1}^{4m+3}\big(\sgn(\CC'_i(\ZZ)-\CC'_{i'}(\ZZ'))\big)\ge 16-9=7>0$, and therefore $\ZZ$ beats $\ZZ'$ on $\fCC'$. 
\subqed
\end{proof}

\begin{prop}
\label{prop:dice_reduce_strict_dice:edge_vs_original}
    Let $\XX,\XX'\in\bits{n}$, and denote their corresponding original vertices $\ZZ=0^{n+2}\XX$ and $\ZZ'=0^{n+2}\XX'$. Let $\ZZ_{\XX,\XX'}\in\bits{2n+2}$ be an edge-vertex corresponding to the edge $(\XX,\XX')$. Then:
    \begin{enumerate}
        \item If $\XX$ ties with $\XX'$ on $\fCC$, then $\ZZ_{\XX,\XX'}$ beats $\ZZ$ and $\ZZ'$ on $\fCC'$.
        \item If $\XX$ beats $\XX'$ on $\fCC$, then $\ZZ$ beats $\ZZ_{\XX,\XX'}$ on $\fCC'$.
    \end{enumerate}
\end{prop}

\begin{proof}
We have:
\begin{equation}
\label{eq:prop_edge_vs_original}
\forall i,j\in [3]\;\; \CC'_{4m+i}(\ZZ_{\XX,\XX'})>\CC'_{4m+j}(\ZZ)   
\end{equation}

For the first part, assume $\XX$ ties with $\XX'$ on $\fCC$. Then $\sum_{i=1}^{m}\sum_{i'=1}^{m}\big(\sgn(\CC_i(\XX)-\CC_{i'}(\XX'))\big)=0$. 
Hence, by \Cref{eq:sum_16_edge_original} we have
$\sum_{i=1}^{4m}\sum_{i'=1}^{4m}\big(\sgn(\CC'_i(\ZZ)-\CC'_{i'}(\ZZ_{\XX,\XX'}))\big)=0$.
Hence, by \Cref{eq:dice_reduce_strict_dice:sum_separation,eq:prop_edge_vs_original}, we have that $\sum_{i=1}^{4m+3}\sum_{i'=1}^{4m+3}\big(\sgn(\CC'_i(\ZZ)-\CC'_{i'}(\ZZ'))\big)=0-9<0$ and therefore $\ZZ_{\XX,\XX'}$ beats $\ZZ$ (and $\ZZ'$, similarly) on $\fCC'$.

For the second part, assume $\XX$ beats $\XX'$ on $\fCC$. 
Then $\sum_{i=1}^{m}\sum_{i'=1}^{m}\big(\sgn(\CC_i(\XX)-\CC_{i'}(\XX'))\big)\ge 1$ and $DL_{\XX,\XX'}=\XX'$. Thus, by \Cref{eq:sum_16_edge_original} we have that
$\sum_{i=1}^{4m}\sum_{j=1}^{4m}\big(\sgn(\CC'_i(\ZZ)-\CC'_j(\ZZ_{\XX,\XX'}))\big)\geq 16$, and by \Cref{eq:dice_reduce_strict_dice:sum_separation} we have:
\[\sum_{i=1}^{4m+3}\sum_{j=1}^{4m+3}\big(\sgn(\CC'_i(\ZZ)-\CC'_j(\ZZ_{\XX,\XX'}))\big)\geq 16-9=7>0\]
and therefore $\ZZ$ beats $\ZZ_{\XX,\XX'}$. 
\subqed
\end{proof}

\begin{prop}
\label{prop:dice_reduce_strict_dice:disjoint_edge_vs_original}
    Let $\XX\in\bits{n}$, and denote the corresponding original vertex $\ZZ=0^{n+2}\XX$. Let $\ZZ_{\XX',\XX''}\in\bits{2n+2}$ be an edge-vertex corresponding to the edge $(\XX',\XX'')$, where $\XX\notin\{\XX',\XX''\}$. If $\XX$ beats $\XX'$ and $\XX''$ on $\fCC$, then $\ZZ$ beats $\ZZ_{\XX',\XX''}$ on $\fCC'$.
\end{prop}

\begin{proof}
We have that $\XX$ beats $DL_{\XX',\XX''}$ and thus by \Cref{eq:sum_16_edge_original} we have that $\sum_{i=1}^{4m}\sum_{j=1}^{4m}\big(\sgn(\CC'_i(\ZZ)-\CC'_j(\ZZ_{\XX,\XX'}))\big)\ge 16$.
Hence, by \Cref{eq:dice_reduce_strict_dice:sum_separation} we have $\sum_{i=1}^{4m+3}\sum_{j=1}^{4m+3}\big(\sgn(\CC'_i(\ZZ)-\CC'_j(\ZZ_{\XX,\XX'}))\big)\geq 16-9=7>0$,
and therefore $\ZZ$ beats $\ZZ_{\XX',\XX''}$. 
\subqed
\end{proof}

\begin{prop}
\label{prop:dice_reduce_strict_dice:edge_vs_edge}
    Let $\ZZ,\ZZ',\ZZ''\in\bits{2n+2}$ be three edge-vertices corresponding to the same edge, where $\ZZ$ begins with $01$, $\ZZ'$ with $10$, and $\ZZ''$ with $11$. Then, on $\fCC'$, $\ZZ$ beats $\ZZ'$, $\ZZ'$ beats $\ZZ''$, and $\ZZ''$ beats $\ZZ$.
\end{prop}

\begin{proof}
$\ZZ$, $\ZZ'$ and $\ZZ''$ all obtain identical values in all circuits $\CC'_i$ for $i\in[4m]$. Hence, by \Cref{eq:dice_reduce_strict_dice:sum_separation} we have that 
\begin{equation}
\label{eq:first_ckts_equal}
\sum_{i=1}^{4m+3}\sum_{j=1}^{4m+3}\big(\sgn(\CC'_i(\ZZ)-\CC'_j(\ZZ'))\big)=
\sum_{i=4m+1}^{4m+3}\sum_{j=4m+1}^{4m+3}\big(\sgn(\CC'_i(\ZZ)-\CC'_j(\ZZ'))\big)
\end{equation}

By definition we have: 
\begin{itemize}
    \item $\CC'_{4m+1}(\ZZ)=2$, $\CC'_{4m+2}(\ZZ)=4$, and $\CC'_{4m+3}(\ZZ)=9$.
    \item $\CC'_{4m+1}(\ZZ)=1$, $\CC'_{4m+2}(\ZZ)=6$, and $\CC'_{4m+3}(\ZZ)=8$.
    \item $\CC'_{4m+1}(\ZZ)=3$, $\CC'_{4m+2}(\ZZ)=5$, and $\CC'_{4m+3}(\ZZ)=7$.
\end{itemize}
Hence, by \Cref{eq:first_ckts_equal}, it may be verified that on $\fCC'$, $\ZZ$ beats $\ZZ'$, $\ZZ'$ beats $\ZZ''$, and $\ZZ''$ beats $\ZZ$.
\subqed
\end{proof}

Now, suppose $\XX^*\in\bits{n}$ is a solution for \CKTDice{} on $\fCC$, and let $\ZZ^*=0^{n+2} \XX$ ($\ZZ^*$ corresponds to $\XX^*$). 
By \Cref{prop:dice_reduce_strict_dice:valid_vs_invalid} $\ZZ^*$ beats all invalid vertices, by \Cref{prop:dice_reduce_strict_dice:original_vs_original} it beats all other original vertices, by \Cref{prop:dice_reduce_strict_dice:edge_vs_original} it beats all edge-vertices corresponding to it, and by \Cref{prop:dice_reduce_strict_dice:disjoint_edge_vs_original} it beats all edge-vertices not corresponding to it. Hence, $\ZZ^*$ is a solution for \StrictCKTDice{} on $\fCC'$. 

Conversely, assume there exists $\ZZ^*\in\bits{2n+2}$ that beats all other $\ZZ\in\bits{2n+2}$ in $\fCC'$. 
By \Cref{prop:dice_reduce_strict_dice:edge_vs_edge,prop:dice_reduce_strict_dice:valid_vs_invalid} we have that $\ZZ^*$ must be an original vertex.
Thus, let $\XX^*\in\bits{n}$ such that $\ZZ^*=0^{n+2} \XX^*$. 
Let $\XX'\in \bits{n}$, and denote by $\ZZ'=0^{n+2}\XX'$ its corresponding original vertex.
If $\XX'$ beats $\XX^*$ on $\fCC$, then by \Cref{prop:dice_reduce_strict_dice:original_vs_original} we have that $\ZZ'$ beats $\ZZ^*$ on $\fCC'$, a contradiction. 
If $\XX'$ ties with $\XX^*$, then by \Cref{prop:dice_reduce_strict_dice:edge_vs_original} the edge-vertices corresponding to the edge $(\XX',\XX^*)$ beat $\ZZ^*$, a contradiction.
Hence, $\XX^*$ beats $\XX'$, and therefore $\XX^*$ is a solution for \StrictCKTDice{} on $\fCC$.
\end{proof}

\begin{lemma}
\label{lem:strict_dice_no_ties}
Given an instance $\fCC=\langle\CC_1,\dots,\CC_m\rangle$ of \StrictCKTDice{}, with $\CC_i\colon\bits{n}\to\bits{n'}$ for all $i\in[m]$, one can construct in polynomial time an instance $\fCC'=\langle\CC'_1,\dots, \CC'_{2m+1}\rangle$ of \StrictCKTDice{}, with $\CC'_i\colon\bits{n}\to\bits{n'+1}$ for all $i\in[2m+1]$, such that:
\begin{enumerate}
    \item $\CC'$ is a yes-instance if and only if $\CC'$ is a yes-instance, and
    \item for each $\XX',\YY'\in \bits{n}$, $\XX'$ and $\YY'$ do not tie on $\CC'$.
\end{enumerate}
\end{lemma}

\begin{proof}
As in the proof of \Cref{lem:dice_reduce_strict_dice}, circuits $\CC'_1,\dots\CC'_{2m}$ are intuitively meant to multiply the original circuits, to enhance any advantage a vertex had over another vertex in the original game. The last circuit is meant to break ties lexicographically.
Thus, let $i\in[m]$ and $\XX\in\bits{n}$. 
We set $\CC'_{2i}(\XX)=\CC'_{2i-1}(\XX)=\CC_i(\XX)$, and $\CC'_{2m+1}(\XX)=2^{n'}+\XX$.

We observe that circuits $\CC'_1,\dots,\CC'_{2m}$ output values upper bounded by $2^{n'}-1$, while circuit $\CC'_{2m+1}$ outputs values lower bounded by $2^{n'}$. Thus, for all vertices $\XX,\XX'\in\bits{n}$ we have:
\begin{equation}
\label{eq:strict_dice_no_ties:sum_separation}
\sum_{i=1}^{2m+1}\sum_{i'=1}^{2m+1}\big(\sgn(\CC'_i(\XX)-\CC'_{i'}(\XX'))\big)=\sum_{i=1}^{2m}\sum_{{i'}=1}^{2m}\big(\sgn(\CC'_i(\XX)-\CC'_{i'}(\XX'))\big)+
\sgn(\CC'_{2m+1}(\XX)-\CC'_{2m+1}(\XX'))
\end{equation}
Furthermore, we have
\begin{equation}
\label{eq:strict_dice_no_ties:sum_4}
\sum_{i=1}^{2m}\sum_{{i'}=1}^{2m}\big(\sgn(\CC'_i(\XX)-\CC'_{i'}(\XX'))\big)=
\sum_{i=1}^{m}\sum_{j=0}^{1}\sum_{{i'}=1}^{m}\sum_{j'=0}^{1}\big(\sgn(\CC'_{2i-j}(\XX)-\CC'_{2i'-j'}(\XX'))\big)=
\end{equation}
\[
\sum_{i=1}^{m}\sum_{j=0}^{1}\sum_{{i'}=1}^{m}\sum_{j'=0}^{1}\big(\sgn(\CC_i(\XX)-\CC_{i'}(\XX'))\big)=
4\sum_{i=1}^{m}\sum_{{i'}=1}^{m}\big(\sgn(\CC_i(\XX)-\CC_{i'}(\XX'))\big)
\]

Additionally, by definition we have 
\begin{equation}
\label{eq:lex_order}
\sgn(\CC'_{2m+1}(\XX)-\CC'_{2m+1}(\XX'))=\sgn(\XX-\XX')
\end{equation}

We require the following proposition to complete the proof of the lemma.
\begin{prop}
\label{prop:mapping_no_ties}
Let $\XX,\XX'\in\bits{n}$ with $\XX>\XX'$. If one of them beats the other in $\fCC$, then it beats it in $\fCC'$ as well. If they tie, then $\XX$ beats $\XX'$ on $\fCC$.
\end{prop}
\begin{proof}
For the first part, assume without loss of generality that $\XX'$ beats $\XX$ on $\CC$. Then we have that $\sum_{i=1}^{m}\sum_{{i'}=1}^{m}\big(\sgn(\CC_i(\XX')-\CC_{i'}(\XX))\big)\ge 1$, and therefore by \Cref{eq:strict_dice_no_ties:sum_separation,eq:strict_dice_no_ties:sum_4} we have that $\sum_{i=1}^{4m+3}\sum_{i'=1}^{4m+3}\big(\sgn(\CC'_i(\XX')-\CC'_{i'}(\XX))\big)\geq 4-1=3>0$.
Hence, we have that $\XX'$ beats $\XX$ on $\fCC'$.

If $\XX$ and $\XX'$ tie on $\fCC$, then $\sum_{i=1}^{m}\sum_{{i'}=1}^{m}\big(\sgn(\CC_i(\XX')-\CC_{i'}(\XX))\big)=0$. 
Hence, by \Cref{eq:strict_dice_no_ties:sum_separation,eq:strict_dice_no_ties:sum_4,eq:lex_order} we have that $\XX^*$ strictly beats $\XX$ on $\fCC'$.
\subqed
\end{proof}

From \Cref{prop:mapping_no_ties} and by definition of \StrictCKTDice{}, it is immediate that $\XX^*$ is a solution of \StrictCKTDice{} on $\fCC$ if and only if it is a solution for \StrictCKTDice{} on $\fCC'$.
Additionally, by \Cref{prop:mapping_no_ties}, no two string can tie on $\fCC'$, since if one of them beats the other on $\fCC$ then it does so in $\fCC'$ as well, and if they tie on $\fCC$ then the numeric winner beats the other on $\fCC'$.
\end{proof}

\begin{definition}
\label{def:epsilon_pairwise_code}
Let $q,n\in\NN$. A set $S$ of words, each in $[q]^n$, is called an \emph{$\epsilon$-pairwise code} if for every $u,w \in S$ and every $a,b \in [q]$ we have that $|\Prob_{i \sim [n]} (u_i = a \land w_i = b) - \frac{1}{q^2}| \le \epsilon$.
\end{definition}

\begin{definition}
\label{def:poly_constructible_set}
A set $S$ of words, each in $[q]^n$ is called \emph{poly-constructible} if there is an algorithm that, when given an index $t \in [|S|]$ outputs the $t$'th word in $S$ in time that is polynomial in the length of the input ($\log(t)$) and the length of the output ($n \log(q)$).
\end{definition}

\begin{theorem}
\label{thm:epsilon_pairwise_code_exists}
(Amnon Ta-Shma and Noam Ta-Shma, personal communication).
For every $q,T\in\NN$ and $\epsilon>0$, there exists $n=poly(q, \epsilon^{-1}, \log (T))$ and a poly-constructible $\epsilon$-pairwise code $S$ over $[q]^n$  with $|S|=T$. 
\end{theorem}

\begin{proof}
(\textbf{\Cref{thm:all_dice_PCW_complete}}).
By \Cref{lem:gdice_pcw,lem:cktdice_pcw} we have that \CKTDice{} and \GDice{} are \PCW{}-hard, and by \Cref{lem:dice_reduce_strict_dice,lem:gdice_cktdice} we have that they both reduce to \StrictCKTDice{}. Hence, it suffices to prove that \StrictCKTDice{} is in \PCW{}. We construct a polynomial time reduction from \StrictCKTDice{} to \CKTCondorcet{}.

Let $\fCC=\langle\CC_1,\dots,\CC_m\rangle$ be an instance of \StrictCKTDice{}, with $\CC_i\colon\bits{n}\to\bits{n'}$ for all $i\in[m]$.
By \Cref{lem:strict_dice_no_ties}, we may assume without loss of generality that no two dice tie on $\fCC$.
Let $\epsilon<\frac{1}{2m^4}$.
Then, by \Cref{thm:epsilon_pairwise_code_exists}, there exists $m'=poly(m,\epsilon^{-1}, \log(2^n))=poly(m,n)$ and a poly-constructible $\epsilon$-pairwise code $S$ over $[m]^{m'}$ with $|S|=2^n$.
For $\XX\in\bits{n}$, let $S(\XX)$ denote the $\XX$'th word in $S$ (where $\XX$ is interpreted as a binary number), and for each $i\in[m']$ let $S(\XX)_i$ denote the $i$'th element of $S(\XX)$ (recall that $S(\XX)\in [m]^{m'}$, and thus $S(\XX)_i\in[m]$).
We define the instance $\fCC'=\langle\CC'_1,\dots \CC'_{m'}\rangle$ by $\CC'_i(\XX)=\CC_{S(\XX)_i}(\XX)$ for any $i\in[m']$ and $\XX\in\bits{n}$.
We can efficiently construct $\CC'$ as $S$ is poly-constructible.

Let $\XX,\YY\in\bits{n}$.
For any $i\in[m']$, let 
\[D_{\XX,\YY,i}=\sgn(\CC'_i(\XX)-\CC'_i(\YY))=\sgn(\CC_{S(\XX)_i}(\XX)-\CC_{S(\YY)_i}(\YY))\]
For $a,b\in[m]$, define the event 
\[W_{\XX,\YY}^{a,b}:=\CC_{a}(\XX)>\CC_{b}(\YY)\]
and let $A=|\{(a,b)\in [m]\times[m]\colon W_{\XX,\YY}^{a,b}\text{ holds}\}|$ and $B=|\{(a,b)\in [m]\times[m]\colon W_{\YY,\XX}^{a,b}\text{ holds}\}|$.
Suppose $\XX$ beats $\YY$ in the Dice game on $\fCC$. Then we have that 
\begin{equation}
\label{eq:A-B_bound}
A - B \ge 1.
\end{equation}
Furthermore, since $A,B\subseteq [m]\times[m]$ we have 
\begin{equation}
\label{eq:A+B_bound}
A + B \le 2m^2.
\end{equation}
Let 
\[F^S_{\XX,\YY}=\sum_{i=1}^{m'}\sgn(\CC'_i(\XX)-\CC'_i(\YY))=\sum_{i=1}^{m'}D_{\XX,\YY,i}\]
Notice that 
$\XX$ is more popular than $\YY$ in $\fCC'$ if and only if $F^S_{\XX,\YY}>0$. Thus, we want to show that $F^S_{\XX,\YY}>0$. 
First, we calculate the expectation of $D_{\XX,\YY,i}$, where $i$ is drawn uniformly at random from $[m']$. We have:
\begin{equation}
\label{eq:E[D]}
\EX_{i\sim [m']}[D_{\XX,\YY,i}]=1\cdot \Prob_{i\sim [m']}(\CC'_i(\XX)>\CC'_i(\YY)) -1\cdot \Prob_{i\sim [m']}(\CC'_i(\YY)>\CC'_i(\XX))+
0\cdot \Prob_{i\sim [m']}(\CC'_i(\XX)=\CC'_i(\YY))
\end{equation}

We therefore calculate the above probabilities:
\[\Prob_{i\sim [m']}(\CC'_i(\XX)>\CC'_i(\YY))=
\Prob_{i\sim [m']}(\CC_{S(\XX)_i}(\XX)>\CC_{S(\YY)_i}(\YY))=
\sum_{\substack{(a,b) \in [m]\times [m] \\ \text{s.t. } W_{\XX,\YY}^{a,b}}}\Prob_{i\sim [m]}(S(\XX)_i=a\land S(\YY)_i=b)\geq\]
\[\sum_{\substack{(a,b) \in [m]\times [m] \\ \text{s.t. } W_{\XX,\YY}^{a,b}}}(\frac{1}{m^2}-\epsilon)=
A\cdot (\frac{1}{m^2}-\epsilon)
\]

where the inequality follows from the definition of an $\epsilon$-pairwise code. Similarly, we have:

\[\Prob_{i\sim [m']}(\CC'_i(\YY)>\CC'_i(\XX))\le 
B\cdot (\frac{1}{m^2}+\epsilon)
\]

Thus, by \Cref{eq:E[D]} we have:
\[\EX_{i\sim [m']}[D_{\XX,\YY,i}]\ge
A\cdot (\frac{1}{m^2}-\epsilon)-B\cdot (\frac{1}{m^2}+\epsilon)
=\frac{A-B}{m^2}-(A+B)\epsilon\ge_{(*)}
\frac{1}{m^2}-2m^2\epsilon>_{(**)}
\frac{1}{m^2}-\frac{2m^2}{2m^4}=0
\]
where inequality $(*)$ follows from \Cref{eq:A-B_bound,eq:A+B_bound}, and inequality $(**)$ follows from the definition of $\epsilon$.
Now, if $\fCC$ is a yes-instance of \StrictCKTDice{} then there exists a die $\XX^*\in\bits{n}$ that beats any other die $\YY\in\bits{n}$. Hence, we have $F^S_{\XX^*,\YY}>0$ for all $\YY\in\bits{n}$, and thus $\XX^*$ is more popular than any $\YY$ in \StrictCKTDice{} on $\fCC'$, namely $\fCC'$ is a yes-instance of \StrictCKTDice{}.
Conversely, assume $\fCC$ is a no-instance of \StrictCKTDice{}, and let $\XX'\in\bits{n}$ be a die. Then there must exist some $\YY'\in\bits{n}$ that beats $\XX'$ in \StrictCKTDice{} defined on $\fCC$.
Thus, we have $F^S_{\XX',\YY'}<0$, and therefore $\XX'$ is not a strongly popular string in the instance $\fCC'$ of \StrictCKTDice{}. Hence, $\fCC'$ is a no-instance of \StrictCKTDice{}.
\end{proof}
\section{Proofs of Section \ref{sec:PMA} (\PMA{})}
\label{sec:PMA_Appendix}
In this section, we provide all proofs missing from \Cref{sec:PMA}, beginning with \Cref{thm:Emaj_k_pma}.
We first show \PMA{}-completeness of an intermediate problem \EMajBal{}, where the induced bipartite graph is balanced but the threshold $k$ remains part of the input.
We then show that parameterizing the number of disqualified edges $k$ does not affect the complexity, beginning with the case $k=1$.
We use these simpler variants to show the remaining results of the section regarding computational bounds on \PMA{}.

\begin{wbox}
\label{def:EMajBal}
    \EMajBal{}\\ 
    \textbf{Input:} Two Boolean circuits $\CC\colon \bits{2n}\to\bits{n+1}$ and $\VV\colon \bits{n+1}\to\bits{2n}$, and integer $k\ge 1$.\\
    \textbf{Question:}
    $\exists \XX\in\bits{n} \;\; \forall \YY\in\bits{n}\;\; \VV(\CC(\XX\concat \YY))=\XX \YY \;\land\; \CC(\XX\concat \YY)\ge k$?
\end{wbox}

\begin{lemma}
\label{thm:emaj_balanced}
\EMajBal{} is \PMA{}-complete.
\end{lemma}
\begin{proof}
    It is clear that $\EMajBal{}\in\PMA{}$ as \EMajBal{} is a special case of \EMaj{}. 
    To show hardness, we reduce from \EMaj{}.
    We begin with a proof sketch. Consider the bipartite graph $G=(X,Y,E)$ induced by a given \EMaj{} instance, with edges only from $X$ to $Y$. If $|X|=|Y|$ then this is an \EMajBal{} instance and we are done. 
    If $|Y|>|X|$ then we can add ``dummy'' vertices to $X$ with out-degree $0$, to balance the sizes of both parts. Thus, the number of edges in the graph does not increase and therefore the bound on the number of edges is still satisfied, so we obtain a valid \EMajBal{} instance, and winners are preserved.
    If $|Y|<|X|$, we add vertices to $Y$, but this time we must make sure that if a winner $\XX^*$ exists in $G$ then the edges from $\XX^*$ to the new $Y$-vertices are present. To do so, we proceed inductively, adding $1$ bit to the representation of the vertices of $Y$ in each step (and in total $n-m$ steps). We think of each vertex $\YY\in Y$ as ``splitting'' into two vertices $\YY_1=0\YY$ and $\YY_2=1\YY$, and similarly any edge label $l$ is split into $0l$ and $1l$. We then make sure that if $\XX\YY$ was an edge with label $l$, then $\XX\YY_1$ and $\XX\YY_2$ are edges with labels $0l$ and $1l$ respectively. It is not hard to show that solutions are preserved by this reduction. 

    Formally, let $n\neq m$, and suppose we are given an instance $\langle \CC,\VV,k \rangle$ of \EMaj{}, where $\CC$ has $n+m$ inputs and $m+1$ outputs, $\VV$ has $m+1$ inputs and $n+m$ outputs, and $k\ge 1$.
    If $n=m$ then this is also a \EMajBal{} instance and we are done.
    We consider the cases $n<m$ and $m<n$.
    
    First, assume $n<m$. Define circuits $\CC'\colon \bits{2m}\to\bits{m+1}$ and $\VV'\colon\bits{m+1}\to\bits{2m}$ as follows. 
    For $\XX,\YY\in\bits{m}$, where $\XX=\XX_1\XX_2$ with $\XX_1\in\bits{m-n}$ and $\XX_2\in\bits{n}$, we set:
    \begin{equation}
    \CC'(\XX\concat \YY)=
    \begin{cases}
    \CC(\XX_2\concat \YY) & \text{ if } \XX_1=0^{m-n}\\
    0^{m+1} & \text{otherwise}
    \end{cases}
    \end{equation}
    For $\ZZ\in\bits{m+1}$, we set 
    \begin{equation}
        \VV'(\ZZ)=0^{m-n} \VV(\ZZ).
    \end{equation}
    Consider the instance $\langle \CC',\VV',k \rangle$ of \EMajBal{}.
    Let $\XX^*\in\bits{n}$ be a solution for \EMaj{} on $\langle \CC,\VV,k \rangle$, and let $\YY\in\bits{m}$.
    We have $\VV(\CC(\XX^*\concat \YY))=\XX^*\YY$ and $\CC(\XX^*\concat \YY)\ge k$. Consider $0^{m-n}\XX^*$. We have $\CC'(0^{m-n}\XX^*\concat \YY)=\CC(\XX^*\concat \YY)\ge k$. Furthermore, we have $\VV'(\CC'(0^{m-n}\XX^*\concat \YY))=\VV'(\CC(\XX^*\concat \YY))=0^{m-n}\VV(\CC(\XX^*\concat \YY))=0^{m-n}\XX^*\YY$. Hence, $0^{m-n}\XX^*$ is a solution for \EMajBal{} on $\langle \CC',\VV',k \rangle$.

    Conversely, assume $0^{n-m}\XX^*$ is a solution for \EMajBal{} on $\langle \CC',\VV',k \rangle$ (note that any solution must have a prefix of $0^{n-m}$, as otherwise $\CC'$ maps it to $0^{m+1}$ regardless of $\YY$, and $0^{m+1}<k$). 
    Let $\YY\in\bits{m}$. 
    We have $\CC(\XX^*\concat \YY)=\CC'(0^{n-m}\XX^*\concat \YY)\ge k$. 
    Furthermore, we have 
    $0^{n-m}\VV(\CC(\XX^*\concat \YY))=\VV'(\CC(\XX^*\concat \YY))=\VV'(\CC'(0^{n-m}\XX^*\concat \YY))=0^{n-m}\XX^*\YY$, implying that $\VV(\CC(\XX^*\concat \YY))=\XX^*\YY$. Therefore, $\XX^*$ is a solution for \EMaj{} on $\langle \CC,\VV,k \rangle$.

    For the second case, assume $m<n$, and let $m'=m+1$.
    We construct an instance $\langle \CC',\VV',k \rangle$ where $\CC'$ has $n+m'$ inputs and $m'+1$ outputs, and $\VV'$ has $m'+1$ inputs and $n+m'$ outputs, such that solutions are preserved. Applying this method inductively $n-m$ times yields an instance where $n=m'$, namely a \EMajBal{} instance.
    We define $\langle \CC',\VV',k \rangle$ as follows.
    For $\XX\in\bits{n}$ and $\YY'\in\bits{m'}$, where $\YY'=y\YY$ with $y\in\{0,1\}$ and $\YY\in\bits{m}$, we set 
    \begin{equation}
        \CC'(\XX\concat \YY')=y\CC(\XX\concat \YY).
    \end{equation}
    Let $\ZZ'\in\bits{m'+1}$, where $\ZZ'=z\ZZ$ with $z\in\{0,1\}$ and $\ZZ\in\bits{m}$. Denote $\XX\YY=\VV(\ZZ)$, where $\XX\in\bits{n}$ and $\YY\in\bits{m}$. We then set
    \begin{equation}
        \VV'(\ZZ')=\XX z\YY.
    \end{equation}

    Let $\XX^*\in\bits{n}$ be a solution for \EMaj{} on $\langle \CC,\VV,k \rangle$, and let $\YY'\in\bits{m'}$, where $\YY'=y\YY$ with $y\in\{0,1\}$ and $\YY\in\bits{m}$.
    We have $\VV(\CC(\XX^*\concat \YY))=\XX^*\YY$ and $\CC(\XX^*\concat \YY)\ge k$. Thus, we have that \[\VV'(\CC'(\XX^*\concat \YY'))=\VV'(y\CC(\XX^*\concat \YY))=\XX^*y\YY=\XX^*\YY'\]
    and also $\CC'(\XX^*\concat \YY')=y\CC(\XX^*\concat \YY)\ge 0\CC(\XX^*\concat \YY)\ge k$. Hence, $\XX^*$ is a solution for \EMaj{} on $\langle \CC',\VV',k \rangle$.

    Conversely, assume $\XX^*$ is a solution for \EMaj{} on $\langle \CC',\VV',k \rangle$. Let $\YY\in\bits{m}$. 
    We have $0\CC(\XX^*\concat \YY)=\CC'(\XX^*,0\YY)\ge k$, which implies that $\CC(\XX^*\concat \YY)\ge k$. 
    Furthermore, we have $\VV'(\CC'(\XX^*,0\YY))=\XX^*0\YY$, which, by definition of $\VV'$, implies that $\VV(\CC(\XX^*\concat \YY))=\XX^*\YY$. 
    Thus, $\XX^*$ is a solution for \EMaj{} on $\langle \CC,\VV,k \rangle$.
\end{proof}

\begin{lemma}
\label{lem:Emaj_1_pma}
$\kEMajBal{1}$ is \PMA{}-complete.
\end{lemma}

\begin{proof}
Clearly $\kEMajBal{1}$ is in \PMA{}, as an instance $\langle \CC,\VV\rangle$ can be reduced trivially to an instance $\langle \CC,\VV, 1\rangle$ of \EMajBal{}.
For hardness, we show that $\EMajBal{}\reduce\kEMajBal{1}$. 
Let $\langle \CC,\VV, k\rangle$ be an instance of \EMajBal{} with $\CC\colon\bits{2n}\to\bits{n+1}$, $\VV\colon\bits{n+1}\to\bits{2n}$, and $k\ge 1$.
Construct the instance $\langle \CC',\VV\rangle$ of $\kEMajBal{1}$, where $\CC'\colon\bits{2n}\to\bits{n+1}$ is defined as follows.
For any $\XX,\YY\in\bits{n}$, let
\[
\CC'(\XX\concat \YY)=
\begin{cases}
\CC(\XX\concat \YY) & \text{ if } \CC(\XX\concat \YY)\ge k\\
0^{n+1} & \text{ if } \CC(\XX\concat \YY) < k
\end{cases}
\]
Now, let $\XX^*$ be a solution for \EMajBal{} on $\langle \CC,\VV, k\rangle$. 
Then for all $\YY\in\bits{n}$ we have $\CC(\XX^*\concat \YY)\ge k$ and $\VV(\CC(\XX^*\concat \YY))=\XX^*\YY$. 
Therefore, we have $\CC'(\XX^*\concat \YY)=\CC(\XX^*\concat \YY)\ge k\ge 1$ and $\VV(\CC'(\XX^*\concat \YY))=\VV(\CC(\XX^*\concat \YY))=\XX^*\YY$.
Thus, $\XX^*$ is a solution for $\kEMajBal{1}$ on $\langle \CC',\VV\rangle$.

Conversely, let $\XX^*$ be a solution for $\kEMajBal{1}$ on $\langle \CC',\VV\rangle$.
Then for all $\YY\in\bits{n}$ we have that $\CC'(\XX^*\concat \YY)\ge 1$ and $\VV'(\CC'(\XX^*\concat \YY))=\XX^*\YY$.
Let $\YY\in\bits{n}$. 
If $\CC(\XX^*\concat \YY)<k$ then $\CC'(\XX^*\concat \YY)=0^{n+1}$, a contradiction to the above. 
Thus $\CC(\XX^*\concat \YY)\ge k$.
Hence, by definition of $\CC'$ we have that $\VV(\CC(\XX^*\concat \YY))=\VV(\CC'(\XX^*\concat \YY))=\XX^*\YY$.
We conclude that $\XX^*$ is a solution for \EMajBal{} on $\langle \CC,\VV,k\rangle$.
\end{proof}

We can now prove \Cref{thm:Emaj_k_pma}.
\begin{proof}
(\textbf{\Cref{thm:Emaj_k_pma}}).
Since the case $k=1$ was settled in \Cref{lem:Emaj_1_pma}, let $k\ge 2$.
Clearly, \PMA{} contains $\kEMajBal{k}$, as an instance $\langle \CC,\VV\rangle$ can be reduced trivially to an instance $\langle \CC,\VV, k\rangle$ of \EMajBal{}.
To show \PMA{}-hardness, we prove that $\kEMajBal{1}\reduce\kEMajBal{k}$, which suffices by \Cref{lem:Emaj_1_pma}.
Let $I=\langle \CC,\VV\rangle$ with $\CC\colon\bits{2n}\to\bits{n+1}$ and $\VV\colon\bits{n+1}\to\bits{2n}$.
Without loss of generality assume $k\le 2^n$ (otherwise we can determine whether it is a yes- or no-instance in constant time, and construct a fixed yes- or no-instance of \EMajBal{} accordingly).
Furthermore, assume that $k$ is a power of two (otherwise, let $k'$ be the smallest power of two greater than $k$; the proof holds for $k'$, and  $\kEMajBal{k'}$ can be reduced to $\kEMajBal{k}$ analogously to the proof of \Cref{lem:Emaj_1_pma}).

To provide intuition, we begin with a proof sketch for the case $k=2$.
We define the instance $I'=\langle \CC',\VV'\rangle$ where $\CC'\colon\bits{2(n+1)}\to\bits{n+2}$ and $\VV'\colon\bits{n+2}\to\bits{2(n+1)}$. 
Since $\CC'$ has one additional output bit compared with $\CC$, its output space is twice as large, namely there are twice as many edge labels. For a label $l$ in $I$ we think of $0l$ and $1l$ as its copies in $I'$. 
Similarly, we can think of each edge as being duplicated: 
If $e:=\XX\YY\in\bits{2n}$ is an edge in $I$ (with $\XX,\YY\in\bits{n}$), then in $I'$ we consider $e':=0\XX0\YY\in\bits{2(n+1)}$ and $e'':=0\XX1\YY\in\bits{2(n+1)}$ as edges corresponding to $e$ (we ignore the edge if $\XX$ is prefixed with $1$ rather than $0$, but we allow both prefixes for $\YY$, and so we have two copies of $e$ and not four).

We wish to define $I'$ so that $e$ is valid if and only if $e'$ and $e''$ are.
Thus, consider a label $l\in\bits{n+1}$ of some edge $e$ in $I$, and its copies $0l$ and $1l$ in $I'$. 
When $l\ge 2$ (interpreted as a binary number), we have that $1l\ge 2$ and $0l\ge 2$ (again, we interpret $0l$ and $1l$ as binary numbers, obtained by prefixing $l$ with the bits $0$ and $1$ respectively). Therefore, both copies are within the acceptable range of labels of $\kEMajBal{2}$, and we can set $\CC'(e')=0l$ and $\CC'(e'')=1l$ (and define $\VV'$ accordingly).

When $l=0^{n+1}$, we do not need to copy the label, as the edge $e$ is invalid in $I$, so it suffices to set $\CC'(e')=\CC'(e'')=0^{n+2}$ to ensure $e'$ and $e''$ are invalid in $I'$. Thus, the label $1l=10^{n+1}$ is currently available for future use.

However, when $l=1$ (or in the general case, when $1\le l<k$), the label is valid in $I$, while not all of its copies are within the acceptable range in $I'$: We have that $0l<2$ (prefixing with $0$ the binary string representing the integer $l$ still yields the same integer), while $1l$ is valid (prefixing $l$ with $1$ yields the integer $2^{n+1}+1>2$).
Therefore, we must find a valid label to replace $0l$. 
To do so, we exploit the fact that the copy $1\concat 0^{n+1}$ (representing the integer $2^n\ge 2$) of the label $0^{n+1}$ is available to replace $0l$.

It is not difficult to generalize this idea to any $k\ge 2$ value. 
Instead of adding one bit to the label space, we add $\log(k)$ bits, thus creating $k$ copies of each edge and each label.
As before, the copies of an invalid edge of $I$ can all be mapped by $\CC'$ to $0^{n+\log(k)}$ to guarantee they remain invalid.
Hence, we have $k-1$ copies of $0^{n+1}$ (all of which encode integers greater or equal to $2^{n+1}\ge k$) which can be used to replace the smallest copy of each label $l\in\bits{n+1}$ with $1\ge l<k$ (the smallest copy of $l$ is always $0^{\log(k)}l$, which encodes an integer equal to $l$, and is thus invalid in $I'$).

We proceed with the formal construction of $I'=\langle \CC',\VV' \rangle$.
Let $k'=\log(k)$ (recall that $k$ is a power of two).
We define two circuits $\CC'\colon \bits{2(n+k')}\to\bits{n+k'+1}$ and $\VV'\colon \bits{n+k'+1}\to\bits{2(n+k')}$. 
Let $\XX',\YY'\in\bits{n+k'}$. Denote $\XX'=\XX'_1\XX'_2$ with $\XX'_1\in\bits{k'}$ and $\XX'_2\in\bits{n}$, and $\YY'=\YY'_1\YY'_2$ with $\YY'_1\in\bits{k'}$ and $\YY'_2\in\bits{n}$.
To define the behavior of the circuits, we make a case distinction based on the values of $\XX'_1$, $\YY'_1$, and $\CC(\XX'_2\concat \YY'_2)$.

First, if $(\XX'_1\ne 0^{k'}) \lor (\CC(\XX'_2\concat \YY'_2)=0^{n+1}) \lor (\VV(\CC(\XX'_2\concat \YY'_2))\ne \XX'_2\YY'_2)$, we let $\CC'(\XX'\concat \YY')=0^{n+k'+1}$. The first condition ensures that only strings of the form $0^{k'}\XX'_2$ can be candidates for winners, while the second and third ensure that copies of invalid edges remain invalid in the reduced instance.

Now, assuming $(\XX'_1=0^{k'}) \wedge (\CC(\XX'_2\concat \YY'_2)\ne 0^{n+1}) \wedge (\VV(\CC(\XX'_2\concat \YY'_2)) = \XX'_2\YY'_2)$, let $l:=\CC(\XX'_2\concat \YY'_2)$ and $i:=\YY'_1$.
We define $\CC'$ formally by case analysis.
For the first case, assume $l\ge k$. Then we let $\CC'(\XX'\concat \YY')=i\concat l$ (where now we interpret $i$ and $l$ as bit strings and concatenate them, obtaining the $i$-th copy of label $l$).

For the second case, assume $1\le l<k$.
Then, if $i\ne 0^{k'}$ we let $\CC'(\XX'\concat \YY')=i\concat l$ as before. 
If $i=0^{k'}$, let $l'$ be the $k'$-bit representation of the binary number $l$ ($k'$ bits suffice, since $l<k=2^{k'}$). Then we set $\CC'(\XX'\concat \YY')=l'\concat 0^{n+1}$.
Namely, it is the $l'$-th copy of $0^{n+1}$ (notice how $l'$ now defines the copy index instead of the label, since the $0$-th copy of the label $l$ is invalid, and so we need to refer to the unused copies of $0^{n+1}$).

The circuit $\VV'$ will be defined as follows. On input $i\concat l$, where $i\in\bits{k'}$ and $l\in\bits{n+1}$, consider the following cases.
If $l\ne 0^{n+1}$ (i.e., $i\concat l$ is not a copy of $0^{n+1}$), let $\XX\YY:=\VV(l)$ with $\XX,\YY\in\bits{n}$. We then define $\VV'(i\concat l)=0^{k'}\XX i\YY$.
If $l=0^{n+1}$ (i.e., $i\concat l$ is a copy of $0^{n+1}$), let $i'$ be the $(n+1)$-bit representation of the integer $i$, and let $\XX\YY:=\VV(i')$ where $\XX,\YY\in\bits{n}$.
Then, define $\VV'(i\concat l)=0^{k'}\XX0^{k'}\YY$ (this is the smallest copy of the edge $\XX\YY$).

We proceed with proving correctness of the reduction.
Suppose $\XX^*\in\bits{n}$ is a solution for the instance $I$ of $\kEMajBal{1}$. Then for all $\YY\in\bits{n}$ we have $\CC(\XX^*\concat \YY)>0$ and $\VV(\CC(\XX^*\concat \YY))=\XX^*\YY$.
We wish to show that $0^{k'}\XX^*\in\bits{n+k'}$ is a solution for $I'$.
Let $\YY'_1\concat\YY'_2\in\bits{n+k'}$, where $\YY'_1\in\bits{k'}$ and $\YY'_2\in\bits{n}$, and denote $l:=\CC(\XX^*\concat \YY'_2)>0$.
We make the following case distinction.

If $l\ge k$, then by definition we have $\CC'(0^{k'}\XX^*\YY'_1\YY'_2)=\YY'_1\concat l\ge k$, and furthermore $\VV'(\CC'(0^{k'}\XX^*\YY'_1\YY'_2))=\VV'(\YY'_1\concat l)=0^{k'}\XX^*\YY'_1\YY'_2$, and thus $0^{k'}\XX^*$ is a solution for $\kEMajBal{k}$ on $\langle\CC',\VV'\rangle$.

If $1\le l<k$ and $\YY'_1\ne 0^{k'}$, then we have $\CC'(0^{k'}\XX^*\concat \YY'_1\YY'_2)=\YY'_1\concat l\ge 2^{n+1}\ge k$ (the prefix $\YY'_1$ ensures the binary interpretation of this string is at least $2^{n+1}$, and we assumed $k\le 2^n$).
Furthermore, we have that $\VV'(\CC'(0^{k'}\XX^*\concat \YY'_1\YY'_2))=\VV'(\YY'_1\concat l)=0^{k'}\XX^*\YY'_1\YY'_2$.

If $1\le l<k$ and $\YY'_1=0^{k'}$, let $l'$ be the $k'$-bit representation of the binary number $l$. 
Then by definition we have $\CC'(0^{k'}\XX^*\concat \YY'_1\YY'_2)=l'\concat 0^{n+1}\ge 2^{n+1}\ge k$ (since $l'>0$, the prefix $l'$ ensures the binary interpretation of this string is at least $2^{n+1}$).
Furthermore, we have that $\VV'(\CC'(0^{k'}\XX^*\concat \YY'_1\YY'_2))=\VV'(l'\concat 0^{n+1})=0^{k'}\XX^*0^{k'}\YY'_2=0^{k'}\XX^*\YY'_1\YY'_2$.

Hence, we have that $0^{k'}\XX^*$ is a solution for $\kEMajBal{k}$ on $\langle\CC',\VV'\rangle$.

Conversely, suppose $0^{k'}\XX^*\in\bits{n+k'}$ is a solution for $\kEMajBal{k}$ on $\langle\CC',\VV'\rangle$ (a solution must have a $0^{k'}$ prefix as otherwise $\CC'$ maps it to $0^{n+k'+1}$ with any $\YY'\in\bits{n+k'}$).
Let $\YY\in\bits{n}$, and denote $l:=\CC(\XX^*\concat \YY)$.
If $l=0^{n+1}$, then clearly $\CC'(0^{k'}\XX^*\concat \YY'\YY)=0^{n+k'+1}$ for any $\YY'\in\bits{k'}$, a contradiction to $0^{k'}\XX^*$ being a solution for $\langle\CC',\VV'\rangle$.
Hence, $l\ne 0^{n+1}$.

Consider the string $0^{k'}\YY$.
We make the following case distinction.

If $l\ge k$ then we have that 
\begin{equation}
\label{eq1:thm_Emaj_k_pma}
0^{k'}\XX^*0^{k'}\YY=\VV'(\CC'(0^{k'}\XX^*,0^{k'}\YY))=\VV'(0^{k'}\concat l)    
\end{equation}
(the former equality is due to the definition of a solution, and the latter by definition of $\CC'$).
On the other hand, denoting $\XX_l\YY_l=\VV(l)$ with $\XX_l,\YY_l\in\bits{n}$, by definition of $\VV'$ we have that
\begin{equation}
\label{eq2:thm_Emaj_k_pma}
    \VV'(0^{k'}\concat l)=0^{k'}\XX_l0^{k'}\YY_l
\end{equation}
By \Cref{eq1:thm_Emaj_k_pma,eq2:thm_Emaj_k_pma} we have that $\XX_l=\XX^*$ and $\YY_l=\YY$.
Hence, $\XX^*\YY=\XX_l\YY_l=\VV(l)=\VV(\CC(\XX^*\concat \YY))$. 

If $l<k$, let $l'$ be the $k'$-bit representation of the binary number $l$.
By similar reasoning to \Cref{eq1:thm_Emaj_k_pma}, we have that
\begin{equation}
\label{eq3:thm_Emaj_k_pma}
0^{k'}\XX^*0^{k'}\YY=\VV'(\CC'(0^{k'}\XX^*,0^{k'}\YY))=\VV'(l'\concat 0^{n+1})    
\end{equation}
On the other hand, denoting $\XX_l\YY_l=\VV(l)$ with $\XX_l,\YY_l\in\bits{n}$, by definition we have that 
\begin{equation}
\label{eq4:thm_Emaj_k_pma}
\VV'(l'\concat 0^{n+1})=0^{k'}\XX_l0^{k'}\YY_l
\end{equation}
By \Cref{eq3:thm_Emaj_k_pma,eq4:thm_Emaj_k_pma} we have that $\XX_l=\XX^*$ and $\YY_l=\YY$.
Hence, $\XX^*\YY=\XX_l\YY_l=\VV(l)=\VV(\CC(\XX^*\concat \YY))$. 

Therefore, $\XX^*$ is a solution for $\langle\CC,\VV\rangle$.
Therefore, the reduction preserves satisfiability, completing the proof of \PMA{}-hardness.
\end{proof}

\begin{proof}
(\textbf{\Cref{thm:EmajSet_pma}}).
Showing \PMA{}-hardness of \EMajSet{} is straightforward, as it can be reduced from $\kEMajBal{1}$ trivially by specifying the set $\{0\}$ of invalid edge labels, while keeping the Boolean circuit untouched.
For inclusion in \PMA{}, we now show that $\EMajSet{}\reduce\EMajBal{}$.
We begin by introducing the following notation.
For a set $T\subseteq \NN_0$ and an element $t\in T$, denote by $idx_T(t)\in [|T|]_0$ the index of $t$ in the list $ordered(T)$, where $ordered(T)$ is the sequence of elements of $T$ sorted in increasing order.
Conversely, for $i\in [|T|]_0$ denote by $T(i)\in T$ the $i$-th element of $ordered(T)$.
Notice that both $idx_T(t)$ and $T(i)$ can be computed in time polynomial in $|T|$ (and in $\log(t')$ for all $t'\in T$).

Now, let $\langle \CC,\VV,S\rangle$ be an instance of \EMajSet{}, where $\CC$ has $2n$ inputs and $n+1$ outputs, and $\VV$ the opposite.
Denote $k:=|S|$, $K=[k]_0$, $S':=S\bs K$, and $K':=K\bs S$.
Notice that $|S|=|K|$, and thus $|S'|=|K'|$.
We construct the instance $\langle \CC',\VV', k \rangle$ of $\EMajBal{}$, where $\CC'$ and $\VV'$ are defined as follows.
The number of inputs and outputs of $\CC'$ and $\VV'$ match those of $\CC$ and $\VV$ respectively.
For $\XX,\YY\in\bits{n}$, define:
\[
\CC'(\XX\concat \YY)=
\begin{cases}
\CC(\XX\concat \YY) & \text{ if } \CC(\XX\concat \YY)\in S\cap K \text{ or } \CC(\XX\concat \YY)\notin S\cup K \\
idx_{S'}(\CC(\XX\concat \YY)) & \text{ if } \CC(\XX\concat \YY)\in S' \\
S'(\CC(\XX\concat \YY)) & \text{ if } \CC(\XX\concat \YY)\in K'
\end{cases}
\]
For $\ZZ\in\bits{n+1}$, define:
\[
\VV'(\ZZ)=
\begin{cases}
\VV(\ZZ) & \text{ if } \ZZ\in S\cap K \text{ or } \ZZ\notin S\cup K \\
\VV(idx_{S'}(\ZZ)) & \text{ if } \ZZ\in S' \\
\VV(S'(\ZZ)) & \text{ if } \ZZ\in K'
\end{cases}
\]
Suppose $\XX^*\in\bits{n}$ is a solution for \EMajSet{} on $\langle \CC,\VV,S\rangle$. Let $\YY\in\bits{n}$. 
We have that $\VV(\CC(\XX^*\concat \YY))=\XX^*\YY$ and $\CC(\XX^*\concat \YY)\notin S$.
If $\CC(\XX^*\concat \YY)\notin S\cup K$, then $\CC'(\XX^*\concat \YY)\ge k$, and furthermore $\VV'(\CC'(\XX^*\concat \YY))=\VV'(\CC(\XX^*\concat \YY))=\VV(\CC(\XX^*\concat \YY))=\XX^*\YY$.
Otherwise, we have that $\CC(\XX^*\concat \YY)\in K'$, and thus $\CC'(\XX^*\concat \YY)=S'(\CC(\XX^*\concat \YY))\in S'$.
Hence, we have that $\VV'(\CC'(\XX^*\concat \YY))=\VV'(S'(\CC(\XX^*\concat \YY)))=
\VV(idx_{S'}(S'(\CC(\XX^*\concat \YY))))=\VV(\CC(\XX^*\concat \YY))=\XX^*\YY$.
Thus, we have that $\XX^*$ is a solution for \EMajBal{} on $\langle\CC',\VV' \rangle$.

Conversely, suppose $\XX^*\in\bits{n}$ is a solution for \EMajBal{}{} on $\langle \CC',\VV',k\rangle$. Let $\YY\in\bits{n}$.
Then $\VV'(\CC'(\XX^*\concat \YY))=\XX^*\YY$ and $\CC'(\XX^*\concat \YY)\notin K$.
If $\CC(\XX^*\concat \YY)\in S'$ then $\CC'(\XX^*\concat \YY)=idx_{S'}(\CC(\XX^*\concat \YY))\in K$, a contradiction. 
Hence, we have that $\CC(\XX^*\concat \YY)\notin K\cup S'=S$.
Additionally, this implies that $\CC(\XX^*\concat \YY)\notin S\cup K$.
Thus, we have that $\VV(\CC(\XX^*\concat \YY))=\VV'(\CC(\XX^*\concat \YY))=\VV'(\CC'(\XX^*\concat \YY))=\XX^*\YY$.
Therefore, $\XX^*$ is a solution for \EMajSet{} on $\langle\CC,\VV, S\rangle$.
\end{proof}

\begin{proof}
(\textbf{\Cref{thm:PMA_Usig2}}).
It is clear that $\PMA{}\subseteq\Sig{2}$, since the condition of \EMajBal{} can be checked in polynomial time for any pair of strings $\XX$ and $\YY$.
Thus, by \Cref{thm:Emaj_k_pma} it suffices to show that any instance $\langle\CC,\VV\rangle$ of $\kEMajBal{1}$ may have at most one solution. 
Assume towards contradiction there exist two strings $\XX_1,\XX_2\in\bits{n}$ where $\XX_1\neq\XX_2$ such that for all $\YY\in\bits{n}$ and $\XX\in\{\XX_1,\XX_2\}$ we have $\VV(\CC(\XX\concat \YY))=\XX \YY \;\land\; \CC(\XX\concat \YY)\neq 0^{n+1}$. 
Then we have $\CC(\XX_1\concat \YY)\in\bits{n+1}\bs \{0^{n+1}\}$ and $\CC(\XX_2\concat \YY)\in\bits{n+1}\bs \{0^{n+1}\}$ for all $\YY\in\bits{n}$.
Furthermore, since $\XX_1\neq\XX_2$, we have $|\{\XX_1\YY\colon \YY\in\bits{n}\}\cup \{\XX_2\YY\colon \YY\in\bits{n}\}|=2^n+2^n=2^{n+1}$.
Thus, by the Pigeonhole Principle, there exist $\YY_1,\YY_2\in\bits{n}$ such that $\CC(\XX_1\concat \YY_1)=\CC(\XX_2\concat \YY_2)$. 
Therefore, we must have that either $\VV(\CC(\XX_1\concat \YY_1))\neq\XX_1 \YY_1$ or $\VV(\CC(\XX_2\concat \YY_2))\neq\XX_2 \YY_2$, a contradiction to the choice of $\XX_1$ or $\XX_2$.
\end{proof}

\begin{proof}
(\textbf{\Cref{thm:PMA_bounds}}, \SymP{} bound).
To prove this, we show that $\kEMajBal{1}\in\SymP{}$. Let $I=\langle\CC,\VV\rangle$ be an instance of $\kEMajBal{1}$.
For a pair $t=(\XX,\YY)$ where $\XX,\YY\in\bits{n}$, let us say that $t$ is \emph{verified} if $\VV(\CC(\XX\concat \YY))=\XX\concat \YY \;\land\; \CC(\XX\concat \YY)\neq 0^{n+1}$.

Consider the bipartite graph $G=(X,Y,E)$ with $X=Y=\bits{n}$ and $E=\{(\XX,\YY)\in X\times Y\colon (\XX,\YY) \text{ is verified}\}$. Indeed, a vertex has out-degree $2^n$ in $G$ if and only if it is a winner in $I$.
Since $\VV$ has $n+1$ inputs, and any pair $(\XX,\YY)$ that is mapped to $0^{n+1}$ by $\CC$ is not verified by design, there are at most $2^{n+1}-1$ verified pairs; this is also an upper bound on $|E|$. If all vertices in $Y$ have in-degree at least two, then we get $|E|\geq 2^{n+1}$, a contradiction to the above. Hence, there exists a vertex $\YY'\in Y$ with in-degree at most one. 

Now, for $\XX\in X$ and $S\subseteq Y$ with $|S|=2$, define the polynomial-time predicate 
\[P(G,S,\XX)=
\begin{cases}
    1 & \text{ if } \forall \YY\in S\;\;(\XX,\YY)\in E\\
    0 & \text{otherwise}
\end{cases}
\]

If $I$ is a no-instance, namely all vertices in $X$ have out-degree at most $2^n-1$, we show that there exists a set $S=\{\YY_1,\YY_2\}$ where $\YY_1,\YY_2\in\bits{n}$, such that for all $\XX$ we have $P(G,S,\XX)=0$, namely either $(\XX,\YY_1)\notin E$ or $(\XX,\YY_2)\notin E$. Indeed, we can construct such a set $S$ as follows. Take $\YY'$ with $in_G(\YY')\leq 1$ as discussed above. If $in_G(\YY')=1$, namely there exists a vertex $\XX''$ such that $(\XX'',\YY')\in E$, then pick $\YY''\in V$ such that $(\XX'',\YY'')\notin E$ (there must exist such $\YY''$ since $I$ is a no-instance). If $in_G(\YY')=0$ then any arbitrary $\YY''$ will do. The set $S=\{\YY',\YY''\}$ then satisfies the above property.

Conversely, if $I$ is a yes-instance with a solution $\XX^*\in\bits{n}$, then for all sets $S=\{\YY_1,\YY_2\}$ we have that $P(G,S,\XX^*)=1$.
Hence, we conclude that $\kEMajBal{1}\in\SymP{}$.
\end{proof}

For the following, we provide a formal definition of the class \coAM{}.

\begin{definition}
\label{def:coAM}
A language $L\subseteq\Sigma^*$ is in $\coAM{}$ if there is a polynomial-time predicate $P(\cdot,\cdot,\cdot)$ such that for all $w\in\Sigma^*$ we have:
\begin{itemize}
    \item $w\in L\implies \Prob_{\XX}[\exists\YY\;\; P(w,\XX,\YY)=1]\le 1/3$, and
    \item $w\notin L\implies \Prob_{\XX}[\exists\YY\;\; P(w,\XX,\YY)=1]\ge 2/3$
\end{itemize}
where $|\XX|$ and $|\YY|$ are bounded by a polynomial function of $|w|$.
\end{definition}

\begin{proof}
(\textbf{\Cref{thm:PMA_bounds}}, \coAM{} bound).
Let $w=\langle \CC,\VV \rangle$ be an $\kEMajBal{1}$ instance, and let $G=(X,Y,E)$ be the graph induced by $w$, where $X=Y=\bits{n}$ and $E$ is defined by $\CC$ and $\VV$ (as detailed in \Cref{sec:PMA}). 
Let $l=\{l_1,l_2\}\subseteq\bits{n+1}$ be two edge labels, and $\YY=\{\YY_1,\YY_2\}\subseteq \bits{n}$ be two $Y$-vertices of $G$.
Further, let $\XX_1$ and $\XX_2$ be the $X$-vertices that $\VV$ maps $l_1$ and $l_2$ to, respectively; that is, we have $\VV(l_1)=\XX_1\concat \YY'$ and $\VV(l_2)=\XX_2\concat \YY''$ for some (insignificant) $\YY'$ and $\YY''$. 
Define the predicate $P(\cdot,\cdot,\cdot)$ by $P(w,l,\YY)=1$ if and only if $(\XX_i,\YY_i)\notin E$ for each $i\in\{1,2\}$. 

If $w$ is a yes-instance, and $\XX^*$ is its solution, then for a random choice of $l$, with probability at least $3/4$ we have that either $l_1$ or $l_2$ belong to $\XX^*$ (that is, for some $i\in\{1,2\}$ we have $\VV(l_i)=\XX^*\concat\YY^*$ for some $\YY^*$, and $\CC(x^*\concat\YY^*)=l_i$). 
If indeed such a label is chosen (and therefore $\XX^*\in\{\XX_1,\XX_2\}$), then $P(w,l,\YY)=0$ for any choice of $\YY$, since no $Y$-vertex can be found that is not adjacent to $\XX^*$.
Thus, we have $\Prob_{l\in\bits{n+1}}[\exists\YY\;\; P(w,l,\YY)=1]\le 1/4<1/3$.

Conversely, if $w$ is a no-instance, then for any choice of $l=(l_1,l_2)$, both corresponding $x_1$ and $x_2$ have some $y_1$ and $y_2$ such that $(x_1,y_1)\notin E$ and $(x_2,y_2)\notin E$. 
Hence, $\Prob_{l\in\bits{n+1}}[\exists\YY\;\; P(w,l,\YY)=1]=1\ge 2/3$.
\end{proof}

Our next goal is to show a lower bound for \PMA{}, and specifically we show that it contains \coNP{}. To do so, we reduce from the canonical \coNP{} version of SAT, defined as follows.

\begin{wbox}
    \UnSat{}\\ 
    \textbf{Input:} A set $\fXX=\{x_1, \dots,x_n\}$ of Boolean variables and a Boolean formula $\psi(\fXX)$ over $\fXX$.\\
    \textbf{Question:}
    $\forall \XX\in\bits{n}\;\;\psi(\XX)=0$?
\end{wbox}

\begin{proof}
(\textbf{\Cref{thm:PMA_bounds}}, \coNP{} bound).
We show that $\UnSat{}\reduce\kEMajBal{1}$, implying any \coNP{} problem is reducible to \EMaj{} and is thus in \PMA{}.
Given an instance $\psi$ of \UnSat{}, we construct the instance $\langle \CC,\VV\rangle$ of $\kEMajBal{1}$, defined by:

\begin{equation}
\forall \XX,\YY\in\bits{n}\; \CC(\XX\concat \YY)=
\begin{cases}
1 \YY & \text{ if } \XX=0^n \land \psi(\YY)=0\\
0^{n+1} & \text{otherwise}
\end{cases}
\end{equation}

\begin{equation}
\forall \ZZ\in\bits{n+1}, \text{ where } \ZZ=z \ZZ' \text{ with } z\in\{0,1\} \text{ and }\ZZ'\in\bits{n},\;\; \VV(\ZZ)=0^n \ZZ'
\end{equation}

First, suppose $\psi$ is a yes-instance of \UnSat{}. We wish to show that $0^n$ is a solution for $\kEMajBal{1}$ on $\langle \CC,\VV\rangle$. Let $\YY\in\bits{n}$. Since $\psi$ is a yes-instance, we have $\psi(\YY)=0$, and therefore we have $\CC(0^n\concat \YY)=1 \YY\neq0^{n+1}$. Furthermore, we have $\VV(\CC(0^n\concat \YY))=\VV(1 \YY)=0^n \YY$. Therefore, $0^n$ is a solution for $\kEMajBal{1}$ on $\langle \CC,\VV\rangle$ by definition.

Conversely, suppose $\langle \CC,\VV\rangle$ is a yes-instance of $\kEMajBal{1}$, and let $\XX^*$ be its solution. Assume towards contradiction $\psi$ is a no-instance, namely there exists $\YY'\in\bits{n}$ such that $\psi(\YY')=1$. Then we have $\CC(\XX^*\concat \YY')=0^{n+1}$, and therefore $\XX^*$ is not a solution for $\kEMajBal{1}$ on $\langle \CC,\VV\rangle$, a contradiction.
\end{proof}
\section{Proofs of Section \ref{sec:PNP} (\PNP{})}
\label{sec:PNP_Appendix}
In this section, we provide all proofs missing from \Cref{sec:PNP}.
For our reductions, we build on the result of \cite{papadimitriou1984complexity} that determining the existence of a unique optimal Traveling Salesman Problem tour is \PNP{}-complete; this was the first problem shown to be complete for \PNP{}, after which several others have been identified (see, e.g., \cite{krentel1986complexity}).

\begin{wbox}
\UTSP\\
\textbf{Input:} Graph $G$ given as an adjacency matrix.\\
\textbf{Question:} Is there a unique minimal-weight Hamiltonian cycle in $G$?
\end{wbox}

\begin{theorem}
\label{thm:UTSP}
(\cite{papadimitriou1984complexity}).
\UTSP{} is \PNP{}-complete.
\end{theorem}

Before addressing the main results of this section, we define a problem we denote \UOPT{}. This is a more general Boolean-circuit formulation of \UTSP{}, which we show is also \PNP{}-complete. It will be convenient to use this formulation for the reductions throughout this section.

\begin{wbox}
    \UOPT\\
    \textbf{Input:} Boolean circuit $\CC$, with $n$ inputs and $n$ outputs.\\
    \textbf{Question:}
    $\exists \XX\in\bits{n} \;\; \forall \XX'\in\bits{n} \text{ with } \XX'\neq \XX \;\; \CC(\XX)>\CC(\XX')$?
\end{wbox}

\begin{lemma}
\label{lem:UOPT}
\UOPT{} is \PNP{}-complete.
\end{lemma}
\begin{proof}
It is immediate from \Cref{thm:UTSP} that \UOPT{} is \PNP{}-hard: Given a TSP instance we construct a circuit $\CC$ that, given a string representing some permutation of the vertices, calculates the length of the Hamiltonian tour induced by this permutation, and multiplies it by $-1$. Thus, there exists a uniquely minimal-weight TSP tour if and only if there exists a unique $\XX$ maximizing the value $\CC(\XX)$.

As for containment in \PNP{}, we can use an algorithm analogous to the one by \cite{papadimitriou1984complexity} showing \UTSP{} is in \PNP{}:
Use an \NP{} oracle to obtain, with binary search, the optimal value obtained by $\CC$. Use the \NP{} oracle to find a string obtaining the optimal value, and an additional query to check if there is an optimizing string different from the one found. If the first string is the only optimizer, then the answer is yes, otherwise it is no.
\end{proof}

\begin{proof}
(\textbf{\Cref{thm:pnp_results}, $\kCKTCondorcet{2}$}).
To show $\kCKTCondorcet{2}$ is \PNP{}-hard, we reduce from \UOPT{} which, by \Cref{lem:UOPT}, suffices. 
Given a \UOPT{} instance $\CC$, we construct a circuit $\CC'$ that outputs $0$ for any input. Thus, a string $\XX$ is a unique optimizer of $\CC(\XX)$ if and only if it is a Condorcet string in $\langle\CC,\CC'\rangle$ (since the circuit $\CC'$ is always indifferent between any two strings).

To show that $\kCKTCondorcet{2}$ is in \PNP{}, we construct a polynomial-time algorithm $A$ that uses an \NP{} oracle $O$. Given two circuits $\langle \CC_1,\CC_2\rangle$, $A$ works as follows.
\begin{enumerate}
    \item Use $O$ to obtain, with binary search, the optimal values $o_1$ and $o_2$ obtained by $\CC_1$ and $\CC_2$ respectively.
    \item Use $O$ to find $\XX$ such that $\CC_1(\XX)=o_1$ and $\CC_2(\XX)=o_2$, or determine there is none.
    \item Use $O$ to find $\XX'\neq \XX$, such that $\CC_1(\XX')=o_1$ and $\CC_2(\XX')=o_2$, or determine there is none.
    \item If $\XX$ exists but $\XX'$ does not, output $1$, otherwise $0$.        
\end{enumerate}

We wish to prove that $A(\CC_1,\CC_2)=1$ if and only if there exists a Condorcet string $\XX$ in $\langle \CC_1,\CC_2\rangle$.
First, assume $A(\CC_1,\CC_2)=1$. Then there is exactly one string $\XX$ that optimizes both circuits. Let $\XX'\neq \XX$. Since $\XX$ achieves the optimal value in both circuits, we have $\CC_1(\XX)\geq\CC_1(\XX')$ and $\CC_2(\XX)\geq\CC_2(\XX')$. Furthermore, since $\XX$ is the unique string optimizing both circuits, one of these inequalities must be strict. Since this holds for any string $\XX'\neq\XX$, we have that $\XX$ is a Condorcet string.

In the other direction, assume $\XX$ is a Condorcet string. Then it must achieve the optimal value in both circuits, since if there existed a string $\XX'$ with $\CC_1(\XX')>\CC_1(\XX)$ or $\CC_2(\XX')>\CC_2(\XX)$, then $\XX$ is not a Condorcet string, a contradiction. Furthermore, if there exists another string $\XX'$ that optimizes both circuits, then clearly $\XX$ and $\XX'$ tie on both circuits, a contradiction to $\XX$ being a Condorcet string. Hence, $\XX$ must be the only string optimizing both circuits, implying $A(\CC_1,\CC_2)=1$.
\end{proof}

\begin{proof}
(\textbf{\Cref{thm:pnp_results}, \CKTConsensus{}}).
By \Cref{lem:UOPT}, we have that \CKTConsensus{} is \PNP{}-hard because \UOPT{} is a special case of \CKTConsensus{} in which there is only one circuit.

We wish to show that $\CKTConsensus{}\in\PNP{}$. Suppose we are given a \CKTConsensus{} instance $\fCC=\langle\CC_1,\dots,\CC_m\rangle$. By definition, $\XX$ is a solution for $\fCC$ if and only if it is a unique optimizer for all circuits $\CC_i$. 
In particular, a solution would optimize $\CC_1$.
Thus, it suffices to find a string $\XX^*$ optimizing $\CC_1$ and use an additional \NP{} query to check if there is a circuit not optimized by $\XX^*$.
This is formalized in the following algorithm, which uses an \NP{} oracle $O$.
\begin{enumerate}
    \item Use $O$ to find, with binary search, the optimal value of $\CC_1$.
    \item Use $O$ to find a string $\XX^*$ optimizing $\CC_1$.
    \item Use a single call to $O$ to answer the following:
    \[\exists \XX'\ne \XX^*,i\text{ such that } \CC_i(\XX')\geq \CC_i(\XX^*)\]
    \item Output the opposite of the previous query's result.
\end{enumerate}
\end{proof}

\begin{proof}
(\textbf{\Cref{thm:pnp_results}, \SDomStrat{}}).
To show that \SDomStrat{} is \PNP{}-hard, we reduce from $\UOPT$. Given a circuit $\CC$ with $n$ input bits, we construct a circuit $\CC'$ with $2n$ input bits, which on input $(\XX\concat\YY)$ with $\XX,\YY\in\bits{n}$, ignores $\YY$ and simply outputs $\CC(\XX)$. Correctness follows from the definitions.

We now wish to show that $\SDomStrat$ is in \PNP{}. Given a circuit $\CC$ for \SDomStrat, we construct an algorithm that is essentially similar to the one in the \PNP{}-completeness proof of \CKTConsensus{} (\Cref{thm:pnp_results}).
\begin{enumerate}
    \item Fix the string $0^n$, and use an \NP{} oracle to find, with binary search, the optimal value of $\CC(\XX\concat 0^n)$.
    \item Use the oracle to find a string $\XX^*$ optimizing $\CC(\XX\concat 0^n)$.
    \item Use a single call to the oracle to answer the following:
    \[\exists \XX'\ne\XX^*,\YY\text{ such that } \CC(\XX'\concat \YY)\geq \CC(\XX^*\concat \YY).\]
    \item Output the opposite of the previous query's result.
\end{enumerate}
\end{proof}

\begin{proof}
(\textbf{\Cref{thm:pnp_results}, \CktWinThresh{}}).
We begin by showing that \PNP{} contains $\CktWinThresh{}$. Given an instance $\CC\colon\bits{2n}\to\{0,1\}$, call a string $\XX$ an \emph{all-winner} if $\CC(\XX\concat\YY)=1$ for any $\YY\in\bits{n}$, and an \emph{all-loser} if $\CC(\XX\concat\YY)=0$ for any $\YY\in\bits{n}$.
We construct the following algorithm using an \NP{} oracle $O$.
\begin{enumerate}
    \item Using $O$, find with binary search the smallest string $\XX^*\in\bits{n}$ which satisfies:
    \[\forall \XX',\YY\in\bits{n}\;\; \XX'\le\XX^* \lor \CC(\XX'\concat \YY)=0.\]
    Namely, $\XX^*$ is the minimal string such that any string above it is an all-loser. Observe that the largest $n$-bit string, $2^n-1$, trivially satisfies this condition, and therefore $\XX^*$ is well defined.
    \item Use a single query to $O$ asking whether
    $\forall \YY\;\;\CC(\XX^*\concat \YY)=1$; if so answer yes, otherwise no.
\end{enumerate}
We show that the algorithm returns yes if and only if $\CC$ is a yes-instance.
Suppose there exists a solution $\XX$ for the instance $\CC$.
In particular, $\XX$ is the smallest string such that all strings above it are all-losers (there cannot be a smaller string satisfying this, since the solution itself is not an all-loser). 
Hence, the string $\XX^*$ found in Step 1 is the solution for $\CC$. 
Therefore, the condition of Step 2 is satisfied by $\XX^*$ and the algorithm returns yes.

Conversely, assume the algorithm returns yes.
Consider the string $\XX^*$ found in Step 1. All strings above $\XX^*$ are all-losers (by the condition of Step 1), and $\XX^*$ itself is an all-winner (by the condition of Step 2).
Hence, $\XX^*$ is a solution for $\CC$. 

To show \CktWinThresh{} is \PNP{}-hard, we reduce from \UOPT{}.
Given an instance $\CC$ of \UOPT{}, we construct a circuit $\CC'$,  defined for every $\YY,\XX,\YY',\XX'\in\bits{n}$ by
\[\CC'(\YY\XX\concat\YY'\XX')=
\begin{cases}
    1 & \text{ if } \CC(\XX)=\YY \land [(\YY>\YY') \lor (\CC(\XX')\neq \YY')\lor (\YY\XX=\YY'\XX')]\\
    0 & \text{ otherwise}
\end{cases}
\]
We show $\CC$ is a yes-instance of \UOPT{} if and only if $\CC'$ is a yes-instance of \CktWinThresh{}.
Assume $\CC$ is a yes-instance. Then there exists $\XX^*\in\bits{n}$ such that $\CC(\XX^*)>\CC(\XX)$ for any $\XX\neq\XX^*$. Denoting $\YY^*:=\CC(\XX^*)$, we show that $\XX^*\YY^*$ is a solution for $\CC'$. 

Let $\YY',\XX'\in\bits{n}$.
We wish to prove that $\YY^*\XX^*$ is an all-winner, and that if $\YY'\XX'>\YY^*\XX^*$ then $\YY'\XX'$ is an all-loser.
For the first part, we make the following case distinction. 
If $\XX'\YY'=\XX^*\YY^*$, then by definition $\CC'(\YY^*\XX^*\concat \YY'\XX')=1$.
If $\XX'\neq \XX^*$ then observe that we either have that $\CC(\XX')\neq \YY'$ or $\CC(\XX')=\YY'$, where in the latter case we must have $\YY^*>\YY'$ since $\YY^*$ is the uniquely-achieved optimum of $\CC$. Hence, we have that $\CC'(\YY^*\XX^*\concat \YY'\XX')=1$.
If $\XX'=\XX^*$ and $\YY'\ne \YY^*$, then since $\YY^*=\CC(\XX^*)$, we have that $\YY'\ne\CC(\XX^*)=\CC(\XX')$, once again implying that $\CC'(\YY^*\XX^*\concat \YY'\XX')=1$.

For the second part, suppose $\YY'\XX'>\YY^*\XX^*$. Then, in particular we have that $\YY'\geq\YY^*$ (as $\YY'$ and $\YY^*$ comprise of the more significant bits of $\YY'\XX'$ and $\YY^*\XX^*$ respectively). Therefore, since $\YY^*$ is the optimal value of $\CC$, and it is uniquely achieved, we must have $\CC(\XX')\neq \YY'$. Thus, we have $\CC'(\YY'\XX'\concat \YY''\concat \XX'')=0$ for any $\YY'',\XX''\in\bits{n}$, namely $\YY'\XX'$ is an all-loser.
We conclude that $\YY^*\XX^*$ is a solution for \CktWinThresh{} on $\CC$.

Conversely, let $\XX^*,\YY^*\in\bits{n}$ such that $\YY^*\XX^*$ is a solution for \CktWinThresh{} on $\CC'$. In particular, we have that $\CC'(\YY^*\XX^*\concat \YY\XX)=1$ for all $\YY,\XX\in\bits{n}$, implying $\CC(\XX^*)=\YY^*$.
Assume towards contradiction that there exists $\XX'\in\bits{n}$ such that $\YY':=\CC(\XX')\geq\CC(\XX^*)$. Then we have that $\CC'(\YY^*\XX^*\concat \YY'\XX')=0$, a contradiction to the choice of $\YY^*\XX^*$.
Hence, $\XX^*$ is the unique optimizer of $\CC$.
\end{proof}
\section{Proofs of Section \ref{sec:Sig2}}
\label{sec:Sig2_Appendix}
In this section, we provide all proofs missing from \Cref{sec:Sig2}.
We begin with a definition of a canonical \Sig{2} satisfiability problem proved to be \Sig{2}-complete by \cite{stockmeyer1976polynomial}, which we will then use for our reductions.

\begin{wbox}
    \QSAT\\ 
    \textbf{Input:} Two sets $\fXX=\{x_1, \dots,x_n\}$ and $\fYY=\{y_1, \dots,y_n\}$ of Boolean variables and a Boolean formula $\psi(\fXX\concat \fYY)$ over $\fXX\cup \fYY$.\\
    \textbf{Question:}
    $\exists \XX^*\in\bits{n}\;\forall\YY'\in\bits{n}\;\psi(\XX^*\concat \YY')=1$?
\end{wbox}

\begin{proof}
(\textbf{\Cref{thm:CKTUVAL_Sig2}}).
Clearly \CKTUVAL{} is contained in \Sig{2}, as for given strings $\XX$ and $\YY$ we can verify in polynomial time whether they obtain different values when given to the circuit.
To show \CKTUVAL{} is \Sig{2}-hard, we reduce from \QSAT{}. Throughout the proof, we interpret bit string as non-negative binary numbers ranging from $0$ to $2^n-1$, and we denote $\bar{\YY}=2^n-1$ (i.e., the all-ones sequence).
We first provide a short proof sketch. Given a Boolean formula $\psi$ over variables sets $\fXX=\{x_1,\dots,x_n\}$ and $\fYY=\{y_1,\dots,y_n\}$, we construct a circuit $\CC$ with $2n$ input bits, denoted $\XX,\YY\in\bits{n}$. We design $\CC$ such that, if $\XX\neq\XX'$ then $\CC(\XX\concat \YY)\neq\CC(\XX'\concat \YY')$ for all $\YY$ and $\YY'$. We do so by ensuring that different strings $\XX$ are allocated disjoint ranges of values when evaluated by $\CC$. Thus, when focusing on a specific string $\XX'$ we need not worry about collisions with other $\XX$ strings. 
Once separate ranges are guaranteed, we wish to ensure $\XX'$ is a solution for $\psi$ if and only if $\CC(\XX'\concat \bar{\YY})$ is obtained uniquely. To do so, we partition all possible $\YY$ values into pairs $(0,1),(2,3),\dots$, except for $\bar{\YY}$ and $\bar{\YY}-1$, which remain unpaired (and are treated separately, as described in the formal proof). We then ensure that if $\YY$ and $\YY'$ are a pair, then $\CC(\XX\concat \YY)=\CC(\XX\concat \YY')$. More specifically, we ensure that if $\psi(\XX\concat \YY)=\psi(\XX\concat \YY')=0$, then $\CC(\XX\concat \YY)=\CC(\XX\concat \YY')=\CC(\XX\concat \bar{\YY})$, but if $\psi(\XX\concat \YY)=\psi(\XX\concat \YY')=1$, then $\CC(\XX\concat \YY)=\CC(\XX\concat \YY')\ne\CC(\XX\concat \bar{\YY})$. 
It then follows that $\CC(\XX\concat \bar{\YY})$ is uniquely obtained if and only if for all $\YY$ we have $\psi(\XX\concat \YY)=1$ (in particular, strings of the form $\XX\concat\bar{\YY}$ are the only candidates for obtaining unique values, as all other strings are paired---precisely for that purpose).

We proceed with the formal proof.
Throughout the proof, we treat all strings as binary numbers. Thus, e.g., $\XX\concat 0$ would be the concatenation of the bit strings representing the numbers $\XX$ and $0$ (in particular, $0$ is not a single bit here).
For all $\YY\in[2^{n}-2]_0$, we define the following function:
\[
p(\YY)=
\begin{cases}
\YY+1 & \text{ if $\YY$ is even}\\
\YY-1 & \text{ if $\YY$ is odd}
\end{cases}
\]
Thus, for all $\YY\in[2^{n-1}-1]_0$, we have that $p(2\YY)=2\YY+1$ and $p(2\YY+1)=2\YY$, and we say $2\YY$ and $2\YY+1$ are \emph{paired}.
Thus, all $\YY\in [2^n-2]_0$ have a pair.
With $\psi$ and $\bar{\YY}$ defined as above, we construct the following circuit $\CC$ on $2n$ input bits.
For $\XX,\YY\in[2^n]_0$, define:
\[
\CC(\XX\concat \YY)=
\begin{cases}
    2\XX & \text{if }\big((\YY\leq \bar{\YY}-2)\land(\psi(\XX\concat \YY)=\psi(\XX\concat p(\YY))=1)\big) \text{ or}\\
    & \big((\YY=\bar{\YY}-1)\land(\psi(\XX\concat \YY)=\psi(\XX\concat \bar{\YY})=\psi(\XX\concat 0)=\psi(\XX\concat 1)=1)\big)\\
    2\XX+1 & otherwise
\end{cases}
\]
We observe some properties of $\CC$. 
\begin{observation}
\label{y_hat_mapping}
For all $\XX\in[2^n]_0$, we have $\CC(\XX\concat \bar{\YY})=2\XX+1$.
\end{observation}

\begin{observation}
\label{pair_mapping}
For all $\XX\in[2^n]_0$ and $\YY\in [2^n-2]_0$, we have that $\CC(\XX\concat \YY)=\CC(\XX\concat p(\YY))$.
\end{observation}

\begin{observation}
\label{different_x_mapping}
For $\XX_1,\XX_2\in [2^n]_0$, if $\XX_1\neq\XX_2$, then $\CC(\XX_1\concat \YY_1)\neq\CC(\XX_2\concat \YY_2)$ for all $\YY_1$ and $\YY_2$.
\end{observation}

\begin{observation}
\label{y_hat_minus_one_mapping}
For all $\XX\in [2^n]_0$, we have that either $\CC(\XX\concat \bar{\YY}-1)=\CC(\XX\concat \bar{\YY})$ or $\CC(\XX\concat \bar{\YY}-1)=\CC(\XX\concat  0)$.
\end{observation}

Assume there exists $\XX'$ such that for all $\YY$ we have $\psi(\XX'\concat \YY)=1$. Then for all $\YY<\bar{\YY}$ we have that $\CC(\XX'\concat \YY)=2\XX'$, while we have $\CC(\XX'\concat \bar{\YY})=2\XX'+1$ by \Cref{y_hat_mapping}. By \Cref{different_x_mapping}, we have that $\CC(\XX'\concat \bar{\YY})$ is uniquely obtained.

In the other direction, assume there exist $\XX',\YY'\in [2^n]_0$ such that $\CC(\XX'\concat \YY')$ is uniquely obtained. By \Cref{pair_mapping,y_hat_minus_one_mapping}, we must have $\YY'=\bar{\YY}$.
We wish to show that $\psi(\XX'\concat \YY)=1$ for all $\YY\in [2^n]_0$. Assume towards contradiction that there exists $\YY_0\in [2^n]_0$ such that $\psi(\XX'\concat \YY_0)=0$.
If $\YY_0<\bar{\YY}-1$, then we have $\CC(\XX'\concat \YY_0)=2\XX'+1=\CC(\XX'\concat \bar{\YY})=\CC(\XX'\concat \YY')$, a contradiction to the assumption that $\CC(\XX'\concat \YY')$ is obtained uniquely.
Similarly, if $\YY_0\in\{\bar{\YY}-1,\bar{\YY}\}$, then we have $\CC(\XX'\concat \bar{\YY}-1)=\CC(\XX'\concat \bar{\YY})$, a contradiction.
\end{proof}

\begin{proof}
(\textbf{\Cref{thm:CKTPareto_Sig2}}).
Clearly \CKTPareto{} is contained in \Sig{2}, as for given strings $\XX$ and $\YY$ we can verify in polynomial time whether $\XX$ obtains a higher output than $\YY$ in one of the circuits. We reduce from \CKTUVAL{} to show it is \Sig{2}-hard. Given an instance $\CC$ of \CKTUVAL{}, we construct the circuits $\CC_1$ and $\CC_2$, where for all $\XX\in\bits{n}$ we set $\CC_1(\XX)=\CC(\XX)$ and $\CC_2(\XX)=-\CC(\XX)$. If there exists $\XX^*$ such that for all $\YY\ne \XX^*$ we have $\CC(\XX^*)\neq\CC(\YY)$, then for all $\YY\ne \XX^*$ we have that either $\CC_1(\XX^*)>\CC_1(\YY)$ or $\CC_2(\XX^*)>\CC_2(\YY)$.
Therefore, $\langle\CC_1,\CC_2\rangle$ is a yes-instance of \CKTPareto{}.
If there is no such string $\XX^*$, then for all $\XX$ there is some $\YY'$ such that $\CC(\XX)=\CC(\YY')$, implying $\CC_1(\XX)=\CC_1(\YY')$ and $\CC_2(\XX)=\CC_2(\YY')$. Hence, $\langle\CC_1,\CC_2\rangle$ is a no-instance of \CKTPareto{}.
\end{proof}
\section{Proofs of Section \ref{sec:Beyond_sig2}}
\label{sec:Beyond_sig2_Appendix}
In this section, we provide all proofs missing from \Cref{sec:Sig2}.

\begin{proof}
(\textbf{\Cref{thm:wdom_uqsat}}).
Showing that $\UQSAT{}\reduce\WDomStrat{}$ can be obtained via a trivial reduction: Given a Boolean formula $\psi$ on variables $\fXX\cup \fYY$, we construct a Boolean circuit $\CC$ with $2n$ inputs and $1$ output, defined for all $\XX,\YY\in\bits{n}$ by $\CC(\XX\concat \YY)=\psi(\XX\concat \YY)$. 
Let $\XX^*\in\bits{n}$. It holds that 
\[\forall \YY\in\bits{n} \;\psi(\XX^*\concat \YY)=1 \;\land \;\forall \XX'\neq\XX^* \;\exists\YY'\; \psi(\XX'\concat \YY')=0\]
if and only if 
\[\forall \XX',\YY\in\bits{n}\;\CC(\XX^*\concat \YY)\geq\CC(\XX'\concat \YY)\land \forall \XX'\neq\XX^* \; \exists\YY'\;\CC(\XX^*\concat \YY)>\CC(\XX'\concat \YY).\]
Thus, $\XX^*$ is a unique solution of \UQSAT{} on $\psi$ if and only if it is a unique $\WDomStrat{}$ solution on $\CC$.

We now show $\WDomStrat{}\reduce\UQSAT{}$.
In the proof, we treat the instance $\psi$ of \UQSAT{} as a Boolean circuit rather than a formula, since we have from Cook's theorem \cite{cook1971complexity} that we can reduce polynomial-time Boolean circuits to polynomial-size Boolean formulae that preserve solutions.
Given a circuit $\CC$ for \WDomStrat{} with $2n$ inputs $\XX,\YY\in\bits{n}$ and $n$ outputs, we construct a circuit $\psi$ for \UQSAT{} with $4n$ inputs and $1$ output. On inputs $\XX',\XX_1,\XX_2,\YY\in\bits{n}$ we define 
\[
\psi(\XX'\concat \XX_1\concat \XX_2\concat \YY)=
\begin{cases}
1 & \text{if } \XX'=0^n\land \CC(\XX_1\concat \YY)\geq \CC(\XX_2\concat \YY)\\
0 & \text{otherwise.}
\end{cases}
\]
If there exists a solution $\XX_1$ for \WDomStrat{} on $\CC$, then 
$\CC(\XX_1\concat \YY)\geq \CC(\XX_2\concat \YY)$ for all $\XX_2,\YY\in\bits{n}$, and thus $0^n\concat \XX_1$ satisfies $\psi$ for any choice of $\XX_2,\YY$. Further, since \WDomStrat{} is unambiguous, and since any $\XX'\neq 0^n$ yields an output of $0$ from $\psi$, we have that $0^n\concat \XX_1$ is the unique solution for \UQSAT{} on $\psi$.

In the opposite direction, assume there exists a solution $\XX'\concat \XX_1$ for \UQSAT{} on $\psi$. 
We must have that $\XX'=0$ as otherwise $\psi$ outputs $0$ for any choice of $\XX_2,\YY$. 
Additionally, we have that $\CC(\XX_1\concat \YY)\geq \CC(\XX_2\concat \YY)$ for all $\XX_2,\YY\in\bits{n}$, by definition of $\psi$.
Therefore, $\XX_1$ satisfies $\CC$ for any choice of $\XX_2,\YY$. 
Hence, it remains to show $\XX_1$ is unique.
Since $\XX_1$ is a solution for \UQSAT{}, it is by definition the unique string satisfying $\forall \YY\in\bits{n}\;\psi(\XX_1,\YY)=1$. Thus, for any $\XX'_1\neq\XX_1$, there exists $\XX_2,\YY\in\bits{n}$ such that $\psi(0^n\concat \XX'_1\concat \XX_2\concat \YY)=0$, which implies $\CC(\XX'_1\concat \YY)<\CC(\XX_2\concat \YY)$. 
This disqualifies $\XX'_1$ from being another solution for \WDomStrat{} on $\CC$, thereby establishing the uniqueness of $\XX_1$.
\end{proof}

\begin{proof}
(\textbf{\Cref{thm:uqsat_dif2}}).
There exists a solution for \UQSAT{} on formula $\psi$ if and only if there exists $\XX^*$ such that for all $\YY$ we have $\psi(\XX^*\concat \YY)=1$, but there do not exist two such $\XX^*$. Since the former condition is a \Sig{2} statement while the latter is a \Piclass{2} statement, we have that $\UQSAT{}\in\Dif{2}$.
By \Cref{thm:wdom_uqsat}, we have that $\WDomStrat{}\in\Dif{2}$ as well.
\end{proof}

To show \UQSAT{} and \WDomStrat{} are lower bounded by $\Piclass{2}$, we consider the following canonical \Piclass{2} SAT problem \cite{stockmeyer1976polynomial}.

\begin{wbox}
    \coQSAT\\ 
    \textbf{Input:} Two sets $\fXX=\{x_1, \dots,x_n\}$ and $\fYY=\{y_1, \dots,y_n\}$ of Boolean variables and a Boolean formula $\psi(\fXX\concat \fYY)$ over $\fXX\cup \fYY$.\\
    \textbf{Question:}
    $\forall \XX\in\bits{n}\;\exists\YY\in\bits{n}\;\psi(\XX\concat \YY)=0$?
\end{wbox}

\begin{proof}
(\textbf{\Cref{thm:uqsat_pi2}}).
We reduce from \coQSAT{} to \UQSAT{}. Our reduction is analogous to the one showing that \USat{} is \coNP{}-hard \cite{blass1982unique}.
Given a Boolean formula $\psi$ over variables $\fXX=\{x_1,\dots,x_n\},\fYY=\{y_1,\dots,y_n\}$, we construct the formula $\psi'$ on variables $\fXX':=\fXX\cup \{x_0\}$ and $\fYY':=\fYY\cup\{y_0\}$, where $x_0$ and $y_0$ are new variables. 
For all assignments $\XX'=\XX'_0\concat\dots\concat\XX'_n$, $\YY'=\YY'_0\concat\dots\concat\YY'_n$, we denote $\XX=:\XX'_1\concat\dots\concat\XX'_n$ and $\YY=:\YY'_1\concat\dots\concat\YY'_n$, and define:
\[\psi'(\XX'\concat\YY')=(\XX'_0\land\dots\land \XX'_n)\lor(\lnot \XX'_0\land \psi(\XX\concat\YY))\]
Notice that $\YY'_0$ is ignored, and is only there to ensure $|\fXX'|=|\fYY'|$.
Clearly the all-ones assignment to $\fXX'$, denoted $\XX^*$, satisfies $\psi'$ regardless of the assignment to $\fYY'$.
Thus, we are interested in characterizing when it is the \emph{unique} assignment satisfying $\psi'$ for all assignments to $\fYY'$.

Assume $\psi$ is a yes-instance of \coQSAT{}, and let $\bar{\XX}\in\bits{n}$ and $\bar{\XX_0}\in\{0,1\}$ such that $\bar{\XX_0}\concat\bar{\XX}\ne\XX^*$.
Then there exists $\bar{\YY}\in\bits{n}$ such that $\psi(\bar{\XX}\concat \bar{\YY})=0$.
Hence, the term $(\lnot \bar{\XX_0}\land \psi(\bar{\XX}\concat \bar{\YY}))$ evaluates to $0$ by definition of $\bar{\YY}$.
Furthermore, the term $(\bar{\XX_0}\land\dots\land \bar{\XX_n})$ evaluates to $0$, as only the all-ones assignment satisfies it. 
Therefore, $\bar{\XX_0}\concat\bar{\XX}$ does not satisfy $\psi'$ for all assignments to $\fYY'$, implying $\XX^*$ is the unique solution. 

Conversely, assume $\psi'$ is a yes-instance of \UQSAT{}. Then there is a unique $(n+1)$-bit string satisfying $\psi'$ for all assignments to $\fYY'$, and by the above reasoning it must be the all-one assignment, $\XX^*$. 
Let $\bar{\XX}\in\bits{n}$.
Assume towards contradiction that for all $\YY\in\bits{n}$ we have $\psi(\bar{\XX}\concat\YY)=1$.
Then, consider the string $0\concat\bar{\XX}\ne\XX^*$. 
By assumption, we have that $(\lnot 0\land \psi(\bar{\XX}\concat\YY))$ evaluates to $1$ for all $\YY\in\bits{n}$, implying $0\concat\bar{\XX}$ satisfies $\psi'$ for all assignments to $\fYY'$, a contradiction to the uniqueness of $\XX^*$.
Hence, there must exist $\bar{\YY}\in\bits{n}$ such that $\psi(\bar{\XX}\concat\bar{\YY})=0$, and thus $\psi$ is a yes-instance of \coQSAT{}.
\end{proof}

\bibliographystyle{alpha}
\bibliography{USigma2}

\end{document}